\documentclass[a4paper,11pt]{amsart}

\usepackage{amsmath,amssymb,amsthm,graphicx,stmaryrd}
\usepackage{bbm}

\usepackage[utf8]{inputenc}
\usepackage[T1]{fontenc}
\usepackage[swedish, english]{babel}

\usepackage{verbatim,xcolor}
\usepackage[shortlabels]{enumitem}
\usepackage{mathtools,cancel}
\usepackage{MnSymbol}
\usepackage{hyperref}
\usepackage[margin=2.5cm]{geometry}

\usepackage[font=sf]{caption}

\numberwithin{equation}{section}

\usepackage{tikz,tikz-cd,ifthen,pgfmath}
\usetikzlibrary{patterns}
\usetikzlibrary{calc}
\usetikzlibrary{arrows.meta}
\usetikzlibrary{decorations.markings}
\usetikzlibrary{decorations.pathmorphing}

\newcommand{\bb}[1]{\mathbb{#1}} 
\renewcommand{\c}[1]{\mathcal{#1}}
\newcommand{\f}[1]{\mathfrak{#1}}

\renewcommand{\bf}[1]{\mathbf{#1}}

\newcommand{\PP}{\bb{P}}
\newcommand{\EE}{\bb{E}}

\newcommand{\e}[1]{\mathrm e^{#1}}

\renewcommand{\a}{\alpha}

\renewcommand{\k}{\kappa} 
\newcommand{\Om}{\Omega}
\newcommand{\om}{\omega}

\newcommand{\eps}{\varepsilon}

\newcommand{\oo}{\infty}

\newcommand{\bra}[1]{\langle#1|}
\newcommand{\ket}[1]{|#1\rangle}

\newcommand{\lam}{\lambda}
\newcommand{\sm}{\setminus}

\newcommand{\es}{\varnothing}
\newcommand{\se}{\subseteq}
\newcommand{\ul}{\underline}

\newcommand{\ol}{\overline}

\newcommand{\floor}[1]{\lfloor #1\rfloor}

\newcommand{\NN}{\mathbb{N}}

\newcommand{\RR}{\mathbb{R}}

\newcommand{\ZZ}{\mathbb{Z}}

\newcommand{\Tr}{\mathrm{Tr}}

\newcommand{\one}{\hbox{\rm 1\kern-.27em I}}
\newcommand{\Ind}[1]{\one_{\{#1\}}}

\newcommand{\be}{\begin{equation}}
	\newcommand{\ee}{\end{equation}}
\newcommand{\bes}{\begin{equation*}}
	\newcommand{\ees}{\end{equation*}}

\renewcommand{\r}[1]{\mathrm{#1}}
\newcommand{\dd}{\r d}
\newcommand{\out}{\r{out}}
\newcommand{\ex}{\r{ex}}
\newcommand{\ins}{\r{in}}
\newcommand{\vol}{\r{vol}}
\newcommand{\per}{\r{perim}}
\newcommand{\p}{\r{per}}
\newcommand{\w}{\r{w}}

\newcommand{\cross}{\mathchoice
	{\vcenter{\hbox{\includegraphics[width=.8em]{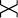}}}}
	{\vcenter{\hbox{\includegraphics[width=.8em]{cross.pdf}}}}
	{\vcenter{\hbox{\includegraphics[width=.6em]{cross.pdf}}}} 
	{\vcenter{\hbox{\includegraphics[width=.5em]{cross.pdf}}}} 
}
\newcommand{\dbar}{\mathchoice
	{\vcenter{\hbox{\includegraphics[width=.8em]{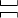}}}}
	{\vcenter{\hbox{\includegraphics[width=.8em]{dbar.pdf}}}}
	{\vcenter{\hbox{\includegraphics[width=.6em]{dbar.pdf}}}} 
	{\vcenter{\hbox{\includegraphics[width=.5em]{dbar.pdf}}}} 
}

\newtheoremstyle{slthm}
{}
{\baselineskip}
{\slshape}
{\parindent}
{\scshape}
{.}
{ }
{}

\theoremstyle{slthm}
\newtheorem{definition}{Definition}[section]
\newtheorem{theorem}[definition]{Theorem}
\newtheorem{proposition}[definition]{Proposition}
\newtheorem{lemma}[definition]{Lemma}
\newtheorem{corollary}[definition]{Corollary}
\newtheorem{remark}[definition]{Remark}

\title[Dimerization]
{Dimerization in $\c O(n)$-invariant quantum spin chains} 
\author{J. E. Bj\"ornberg}
\thanks{JEB: Chalmers University of Technology and University of
	Gothenburg, Sweden. jakob.bjornberg@gu.se}
\author{K. Ryan}
\thanks{KR: TU Wien, Austria. kieran.ryan@tuwien.ac.at}

\date{\today}

\begin{document}

	\begin{abstract}
		We establish dimerization in $\c O(n)$-invariant quantum spin chains
		with large enough $n$, in a much larger region of 
		the phase diagram than previous results. This entails identifying two distinct infinite volume
		ground states which are gapped, are translations by one unit of each other, and
		which both have exponentially decaying correlations.
		We show these states are the limits of finite-volume, finite temperature states
		as the graph tends to $\ZZ$ and temperature tends to zero in any order.
		Our method relies on a probabilistic
		representation of the quantum system in terms of random loops, and an
		adaptation of a method developed for loop $\c O(n)$
		models on the hexagonal lattice by Duminil-Copin, Peled, Samotij and
		Spinka. 
	\end{abstract}
	
	\maketitle	
	\section{Introduction}
	
	The most general $\c O(n)$-invariant quantum spin system with
	pair-interactions has Hamiltonian 
	\be\label{eq:ham-general}
	H_\Lambda=-\sum_{xy\in\c E(\Lambda)}
	\big[
	u T_{xy}+v Q_{xy}
	\big],
	\qquad\text{acting on } (\bb C^n)^{\otimes \Lambda},
	\ee
	where $u,v\in\RR$,
	$\Lambda$ is a finite graph,
	$\c E(\Lambda)$ is its set of nearest-neighbour pairs,
	and the interaction involves the operators
	\be
	T\ket{a,b}=\ket{b,a},
	\qquad
	Q=\frac1{n}\sum_{a,b=1}^n \ket{b,b}\bra{a,a},
	\qquad\text{on } (\bb C^n)^{\otimes 2}.
	\ee
	The behaviour of this model varies widely depending on the
	graph $\Lambda$, the parameters $u,v$, and the value of $n$
	(related to the spin $S$ by $2S+1=n$). The model reduces to
	the spin-$\tfrac{1}{2}$ Heisenberg XXZ model when $n=2$, and
	to the spin-1 bilinear-biquadratic Heisenberg model when
	$n=3$. For $u,v\ge0$, the model has a well-known probabilistic
	representation as a loop model (related to the interchange
	process).

	This paper concerns the model in one dimension with
	$u\ge0, v>0$ and $n$ large, where we prove
	\emph{dimerization} in the ground state: the existence of two
	distinct extremal gapped ground states which are $2\ZZ$-invariant and shifts
	of each other by 1, and
	which have exponentially
	decaying correlations. 
	One can think of the spin at each site on the lattice
	$\ZZ$ binding tighter with one of its neighbours than the other. There
	are two configurations giving such a behaviour for all particles:
	either the particles on even sites are all bound tighter to the
	particle to their right, or their left. \\
	
	Let us review recent precedents for our main result Theorem
	\ref{thm:dimersation}. The case 
	$u=0$, $v>0$ (and any $n$) is special because its loop
	representation coincides with that  of planar FK percolation with
	$q=n^2$, once a scaling limit in one spatial direction is
	taken. This means that the tools available to FK percolation
	are useful also for the loop model, notably the
	FKG inequality. Infinite volume limits of the model under
	``even'' and ``odd'' boundary conditions (corresponding to
	wired and free in FK percolation) are easily shown to
	exist. Aizenman and Nachtergaele \cite{an} proved a dichotomy:
	either these two infinite volume measures differ and
	dimerization occurs, and correlations in the quantum model
	decay exponentially fast, or the measures coincide, and one
	has slow decaying correlations: $\sum_{x\in\ZZ}|x||\langle
	\bf{S}_0\cdot\bf{S}_x\rangle|=\infty$.

	For $u=0$ and all
	$n\ge3$, Aizenman, 
	Duminil-Copin and Warzel \cite{ADCW} proved that the first alternative of the dichotomy holds, using an adaptation of Ray and Spinka's proof of a discontinuous phase transition in FK percolation for $q>4$ \cite{ray-spinka}. Nachtergaele and Ueltschi
	\cite{nacht-uelt} had proved this a few years earlier for $n\geq17$.
	
	Away from the point $u=0$, the model is much harder to
        analyse. The loops in the loop model cross each other, which
        means there is no link to FK percolation and no FKG
        inequality, and so no automatic convergence in infinite volume
        of the ``odd'' and ``even'' measures.
The dichotomy of \cite[Theorem~6.1]{an}
mentioned above is also lacking, as its proof uses FKG. 
	Bj\"ornberg, M\"uhlbacher, Nachtergaele and
	Ueltschi \cite{BMNU} showed that dimerization occurs for $|u|$ small enough (including
	negative $u$) and $n$ large enough, via a cluster-expansion.
	
	Our result proves dimerization for all $u\ge0$, $v>0$ and all
        $n$ large enough, a marked improvement in the $u\neq0$ regime
        on the previous results. Our statements are also stronger than
        those of \cite{BMNU} in the sense that we show that different
        finite volume, finite temperature states converge to the two
        different infinite volume ground states (rather than only
        showing that the limits of certain individual observables
        differ). The region where dimerization is expected is $v>0$
        and $u>-\tfrac{n-2}{2n}v$; our result covers the entire
        subregion where
        $u\ge0$. See Figure \ref{fig:1dphasediagram}.\\

	In a separate paper \cite{BR-exp}, we prove that in $\ZZ^d$ for all
	$d\ge1$, there is  exponential decay of two-point correlations
        in the periodic state at finite
	low temperature, for $v\ge u\geq0$ and $n$  large,
        using different 	techniques.

	\subsection{Main theorem}
	
	To state our result on dimerization, we need some definitions. 
	For positive temperature $T=\frac1\beta$, the Gibbs state 
	$\langle \cdot \rangle_{\Lambda, \beta}$
	in volume $\Lambda$ is the linear map
	$L( (\bb C^n)^{\otimes \Lambda})\to\bb C$
	given by
	\be\label{eq gibbs state}
	\langle A \rangle_{\Lambda, \beta}
	=
	\frac1{Z_{\Lambda,n,\beta}}
	\Tr\big(A\e{-\beta H_\Lambda}\big),
	\qquad\text{where }
	Z_{\Lambda,n,\beta}=\Tr\big(\e{-\beta H_\Lambda}\big).
	\ee
	The ground state is given by replacing the Gibbs factor
	$\e{-\beta H_\Lambda}$ by the projection onto the subspace of the
	lowest eigenvalue of $H_\Lambda$, equivalently it is the limit
	$\beta\to\oo$ of the Gibbs state. 
	Infinite-volume Gibbs- or ground-states can be obtained as
	limits of these as $\Lambda\Uparrow\bb Z^d$
        (provided the limit exists; it is straightforward to show that
        subsequential limits exist, however), or characterized using 
	the Kubo--Martin--Schwinger (KMS) condition.
	
	For a vector
	$\Psi\in(\bb C^n)^{\otimes\Lambda}$, write
	\be\label{eq seeded state}
	\langle A \rangle_{\Lambda, \beta}^{\Psi}
	=
	\frac{
		\langle \Psi | \e{-\frac\beta2 H_{\Lambda}} A
		\e{-\frac\beta2 H_{\Lambda}} | \Psi \rangle
	}
	{
		\langle \Psi | e^{-\beta H_{\Lambda}} | \Psi \rangle
	},
	\ee
	which we refer to as a \emph{seeded} state.
	For $\Lambda=\Lambda_L:=\{-L+1,\dots,L\}\subset\ZZ$, 
	\phantomsection\label{not LL}
	let
	\be\label{eq:Psi}
	\Psi_L
	=
	\sum_{i=1}^L 
	\frac{1}{\sqrt{n}}
	\sum_{a=1}^n
	| a,a \rangle_{-L+2i-1, -L+2i}.
	\ee
	Note that $Q$ is the orthogonal projector on $\Psi_1$.
	By \emph{local operator} we mean 
	a linear operator $A$ on  $(\bb C^n)^{\otimes\Lambda(A)}$
	for a finite set  $\Lambda(A)$; 
	\phantomsection\label{not LA}
	we regard $A$ as an operator on 
	$(\bb C^n)^{\otimes\Lambda}$ for
	any $\Lambda\supseteq \Lambda(A)$ by identifying it with
	$A\otimes\one_{\Lambda\sm\Lambda(A)}$.
	The smallest choice of set $\Lambda(A)$ is called the operator's
	\emph{support}.  
	For $t\in\RR$ and a local operator $A$ (with support in $\Lambda$)
	we write
		$A(t)=e^{-tH_{\Lambda}}Ae^{tH_{\Lambda}}$,
	and for a state
	$\langle\cdot\rangle$ we write $\langle A;B \rangle = \langle AB \rangle
	- \langle A \rangle\langle B \rangle$. 
	We write 
	$\|A\|=\sup_{\|\psi\|=1}|\bra\psi A\ket \psi|$ for the operator norm.

	We consider the model \eqref{eq:ham-general} with $v=1-u$ and
	$u\in[0,1]$, so the Hamiltonian is  
	\be\label{eq:ham}
	H_\Lambda=-\sum_{xy\in\c E(\Lambda)}
	\big[
	u T_{xy}+(1-u) Q_{xy}
	\big].
	\ee
	\begin{theorem}[Dimerization]\label{thm:dimersation}
		Let $d=1$ and fix $u\in[0,1)$.  
		There is $n_0(u)$ such that for $n>n_0$, the following holds:
		\begin{enumerate}[1.,leftmargin=*]
			\item  There  are  two distinct infinite-volume ground-states
			$\langle\cdot\rangle_1$, $\langle\cdot\rangle_2$ for the 
			model \eqref{eq:ham}, such that for $\alpha\in\{1,2\}$,
			\be
			\langle\cdot\rangle_\alpha=
			\lim_{\substack{L\to\infty \\ L\in 2\bb Z+\alpha}} 
			\lim_{\beta\to\infty}
			\langle\cdot\rangle_{\Lambda_L}
			=
			\lim_{\substack{L\to\infty \\ L\in 2\bb Z+\alpha}} 
			\lim_{\beta\to\infty}
			\langle\cdot\rangle_{\Lambda_L}^{\Psi_L}.
			\ee
			Moreover, for both
			$\langle\cdot\rangle_{\Lambda_L}$ and
			$\langle\cdot\rangle_{\Lambda_L}^{\Psi_L}$ the limits $L,\beta\to\oo$
			can also be taken simultaneously, and for 
			$\langle\cdot\rangle_{\Lambda_L}^{\Psi_L}$
			the limits can be taken in any order.
			The states $\langle\cdot\rangle_\alpha$
			are $2\ZZ$-invariant and are translations by one unit of
			each other. 
\item
Correlations decay exponentially in both limiting states: there
  exists $C',C>0$ such that for any local operators $A,B$ with supports
  $\Lambda(A)$, $\Lambda(B)$ respectively, we have for $\alpha\in\{1,2\}$ and
  for all $t\in\RR$,  
  \be
  |\langle A;B(t)\rangle_\alpha|
  \le
  C'\|A\| \|B\|  |\Lambda(A)||\Lambda(B)|\e{- C(d(\Lambda(A), \Lambda(B)) + |t|)},
  \ee
  where $d(\Lambda(A), \Lambda(B))$ denoted the minimal distance
  between points in the supports.
  \item   The ground-states $\langle\cdot\rangle_\alpha$ are gapped:
    there exists $\delta>0$ such that the lowest and next-lowest eigenvalues 
    $e_0(L)$ and $e_1(L)$ of $H_{\Lambda_L}$ satisfy
    $e_1(L)-e_0(L)\geq \delta$ for all large enough $L\in2\bb Z+\alpha$.
  \end{enumerate} 
	\end{theorem}
	
\begin{remark}
  The decay of correlation  implies that the two
  states $\langle\cdot\rangle_\alpha$, $\alpha\in\{1,2\}$,
  are extremal  (i.e.\ pure), see \cite[Corollary~IV.2.3]{israel}.
\end{remark}
	
	\subsection{Background}
	
	The 1-dimensional ground state behaviour of the model
	\eqref{eq:ham-general} is diverse.
	The expected ground state for the model with $n\geq 3$
	is depicted in Figure \ref{fig:1dphasediagram}, and was first studied in \cite{tu-zhang-xiang-spin-1+}.  
	We summarize the heuristics for the phase diagram;
	for further details see \cite{BMNU}, \cite{fath-solyom-spin-1}, \cite{tu-zhang-xiang-spin-1+}.
	\begin{figure}[th]
		\begin{center}
			\includegraphics[scale=1]{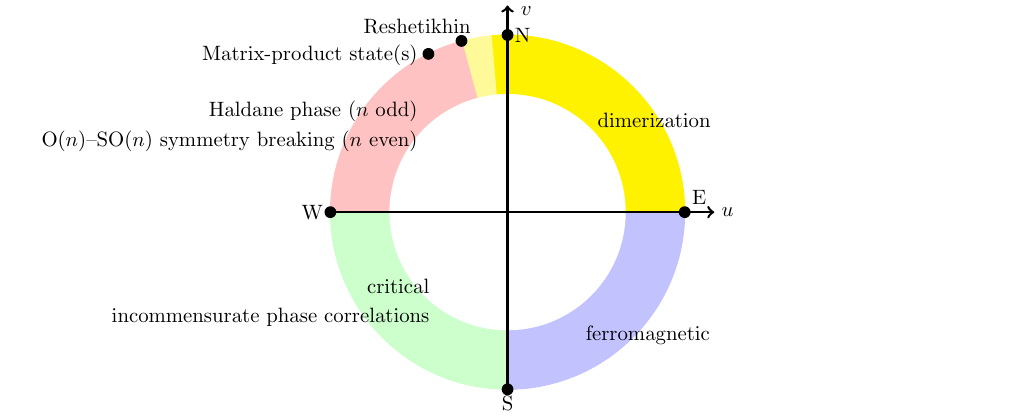}%
		\end{center}
		\caption{Expected ground state phase diagram for the spin chain
			with Hamiltonian \eqref{eq:ham-general} for $n\geq3$.  
			Dimerization is expected in the yellow region, and has been
			established in the darker yellow region, for certain values of $n$.
			In this paper we prove dimerization in the range from N to E
			(E not included), for large enough $n$.  It was
			previously established at N (for $n\geq3$) and in  a
			neighbourhood around  N (for large $n$).}
		\label{fig:1dphasediagram}
	\end{figure}

	The operator
	$T_{x,y}$ can be thought of as ferromagnetic, and $Q_{x,y}$
	anti-ferromagnetic. This leads to the most straightforward part of the
	phase diagram: the southeast quadrant $u\geq0,v\leq0$ (blue in Figure
	\ref{fig:1dphasediagram}, including the point $E$) is ferromagnetic;
	there are many ground states, which  minimize the
	Hamiltonian \eqref{eq:ham-general} term by term.  
	
	In the southwest quadrant $u<0,v<0$ (green in Figure
	\ref{fig:1dphasediagram}), the model is ``critical'', in the
	sense that there should be a unique, gapless ground state,
	with polynomially decaying correlations. Correlations are in
	fact expected to decay with incommensurate phases, that is,
	correlations between the origin and $x$ should decay as
	$|x|^{-r}\cos(\om|x|)$, where $r,\om$ are functions of
	$u,v$. See \cite{fath-solyom-spin-1} and
	\cite{itoi-kato-spin-1-massless}. 
	
	Between the point $W$ and the Reshetikhin point (red in Figure \ref{fig:1dphasediagram}), one expects behaviour dependent on the parity of $n$. For $n$ odd the model is in the Haldane phase: there should be a unique ground state, with a gap and exponentially decaying correlations. For $n$ even, there should be two extremal gapped ground states, which are translations of each other by 1, but which are \textit{not} dimerized states (as thought in some physics literature). Rather, the $O(n)$ symmetry is broken down to $SO(n)$, and the two states are related by a site-transformation of determinant $-1$. This behaviour for $n$ even was recently discovered in the PhD thesis of Ragone \cite{ragone-thesis}, and had not been observed in a model before. He studies the point $v=-2u$, which is a frustration-free point whose ground states are given by matrix product states. For $n=3$, this is the AKLT model \cite{AKLT1, AKLT2}, and for larger $n$ named the $SO(n)$ AKLT model. Ragone studies this point for all $n\ge3$, and proves the behaviour described above. The $SO(n)$ AKLT point should have the same qualitative behaviour as the whole phase between the point $W$ and the Reshetikhin point. General stability results on gapped chains should extend the existence of the gap rigorously to a neighbourhood of the point $v=-2u$, see \cite{nachtergaele-sims-young}.
	
	The point $v=-\tfrac{2n}{n-2}u$ was solved by Reshetikhin \cite{reshetikhin} and shown to have no gap. The special case when $n=3$ is the Takhtajan–Babujian model.

	Finally, as discussed above, between the point $E$ and the Reshetikhin point (yellow in
	Figure \ref{fig:1dphasediagram}) dimerization is expected. One way to distinguish
	dimerization from the $O(n)-SO(n)$ symmetry breaking of the $SO(n)$
	AKLT model for $n$ even is to show that the expectation of some
	two-site observable is different in the two ground states (this will
	be the operator $Q_{x,y}$ defined above); interestingly the two ground
	states of the $SO(n)$ AKLT model for $n$ even cannot be distinguished
	by any two-site observable \cite{ragone-thesis}. 
	
	The spin chains \eqref{eq:ham-general} can be represented using loop
	models, where loops travel along components of $\ZZ\times\RR$, joined
	by links between nearest neighbours.
	The inverse temperature $\beta$ corresponds to the
	height (in the $\RR$ direction) of the finite-volume loop
	model, so studying the infinite volume ground state of the
	spin model amounts to studying the 2D infinite volume limits
	of the loop model.
	The range $u,v\geq0$, which we study in this paper, is special in that 
	the loop model is probabilistic, while for parameters outside that
	range the loop model comes with a signed measure.
	Dimerization in the loop model setting is the existence of two
	distinct infinite volume Gibbs measures, which are
	translations of each other by 1. These are easy to visualise:
	the model prefers many loops, so prefers short loops. The
	shortest possible loops are those which touch only two links
	between the same two nearest neighbours.
	The two Gibbs
	measures each display a unique infinite cluster of short loops
	all lying either on odd edges of $\ZZ$ or all on even
	edges.

	We do not attempt to obtain the optimal value for the threshold 
	$n_0(u)$ as it is clear that our methods give only a very weak upper
	bound.  It has been suggested that the optimal value is $n_0(u)=2$ for
	all $u\in[0,1)$.  Our bound, however, diverges as $u\to1$.
	
	The proof of Theorem \ref{thm:dimersation} is inspired by the proof of a similar result for the loop $O(n)$ model on the hexagonal lattice by Duminil-Copin, Peled, Samotij and Spinka \cite{DCPSS}. Adapting the proof to the setting of the loop model defined from an underlying Poisson process requires significant modifications.\\

	Note that the model \eqref{eq:ham-general} with $n=2$ (spin
	$\frac12$) behaves
	differently in its ground state to $n\ge3$. The model is
	equivalent to the Heisenberg XXZ model. Without loss of
	generality, setting $v=1-u$ and $\Delta=2u-1$, 
	the Hamiltonian is equivalent to:  
	\be\label{eq:ham-xxz}
	H^{\textsc{xxz}}_\Lambda=-\sum_{xy\in\c E(\Lambda)}
	\big[
	S^{(1)}_xS^{(1)}_y + S^{(2)}_xS^{(2)}_y + \Delta S^{(3)}_xS^{(3)}_y
	\big],
	\ee
	where $S^{(1)},S^{(2)},S^{(3)}$ are the usual spin operators.
	Write $\Delta=2u-1$. There are the special cases $\Delta=1$ (the
	Heisenberg ferromagnet), $\Delta=-1$ (the Heisenberg antiferromagnet),
	and $\Delta=0$ (the XY model).  
	
	For $\Delta=2u-1\in[-1,1)$, the model is ``critical'', in the
	sense that there is expected to be a unique, gapless ground
	state with polynomially decaying correlations. The model is
	widely studied using exact solutions methods; see for example
	the textbooks \cite{KBI-book, sutherland-book,
		takahashi-book}. In finite volume, the ground state is
	unique and Lieb, Schultz and Mattis \cite{lieb-schultz-mattis}
	showed that the spectral gap is at most $\mathrm{const.}/L$,
	with $L$ the length of the system. To the authors' knowledge,
	the only rigorous proof of a unique ground state in infinite
	volume is for the XY model ($\Delta=0$) by Araki and Matsui
	\cite{araki-matsui-XY}. However, for the antiferromagnet
	($\Delta=-1$), the result that the loop model has a unique
	infinite volume Gibbs measure is a consequence the same result
	for quantum FK percolation with $q=4$ by Duminil-Copin, Li and
	Manolescu \cite{DC-li-manolescu}. Affleck and Lieb
	\cite{affleck-lieb-XXZ} proved that a unique ground state
	implies there is no spectral gap; their result extends to all
	half-odd-integer spins and is some evidence for the Haldane
	conjecture. Polynomial decay of correlations also implies no
	spectral gap, see for example Problem 6.1.a in
	\cite{tasaki-book}. 
	
	For $\Delta<-1$ the $S^{(3)}$ term dominates and one expects behaviour like the antiferromagnetic Ising model (which is the point $v=-u>0$, or $\Delta=-\infty$). For $\Delta<-1$ and $|\Delta|$ sufficiently large, Matsui \cite{matsui-XXZ-antiferro} proved there are exactly two extremal ground states, and for all $\Delta<-1$, Aizenman, Duminil-Copin and Warzel \cite{ADCW} proved the existence of two distinct ground states (which should be the only two extremal ones) exhibiting Néel order. 
	
	For $\Delta\ge1$ one has ferromagnetic behaviour. For
	$\Delta>1$, the $S^{(3)}$ term dominates once again, and one
	has the behaviour of the ferromagnetic Ising model (which is
	the point $v=-u<0$, or $\Delta=+\infty$). Here there are two
	translation-invariant ground states (all spins up and all
	spins down), and an infinite number of
	non-translation-invariant ground states (all spins left of
	$x\in\ZZ$ up (resp.\ down) and all spins right of $x$ down
	(resp.\ up), known as kink (or anti-kink) states. This was
	proved to be a complete list of all the extremal ground states
	by Matsui \cite{matsui-XXZ-ferro}, a result extended to all
	spins by Koma and Nachtergaele
	\cite{koma-nacht-XXZ-ferro}. The kink and antikink states were
	discovered by Alcaraz, Salinas, 
	and Wreszinski \cite{ASW}, and Gottstein and Werner \cite{gottstein-werner}. For $\Delta=1$ (the Heisenberg ferromagnet), the Ising behaviour disappears and for all spins, all ground states are translation-invariant \cite{koma-nacht-XXZ-ferro} (in fact for spin$-\tfrac{1}{2}$ they are exactly all of the permutation-invariant states). \\

	\subsection{Mirror model}
	
	Our results have analogues in
	a random mirror model, which is a discrete
	version of the probabilistic loop-model which is our main tool
	(Section  \ref{section:prob-rep}).
	We summarize the results for the mirror model
	here, but do not give proofs as they can be obtained through
	straightforward modifications of the arguments in \cite{DCPSS}.
	
	Consider a finite subset $\Lambda$ of $\bb Z^2$, which we think of as
	rotated $45^\circ$.  On each site $x\in\Lambda$, we place either a
	vertical mirror, a horizontal mirror, or no mirror.   A \emph{mirror
		configuration} is thus an element 
	$\xi\in\{\r v,\r h,\varnothing\}^{\Lambda}$.  
	Mirrors are reflective on both sides, so rays of light travelling
	along the edges of the lattice assemble into loops and paths
	(possibly depending on a boundary condition).
	We let $\ell(\xi)$ be the number of loops (or paths) which intersect 
	$\Lambda$.  
	See Figure \ref{fig mirrors}.
	
	\begin{figure}[th]
		\begin{center}
			\includegraphics[scale=1]{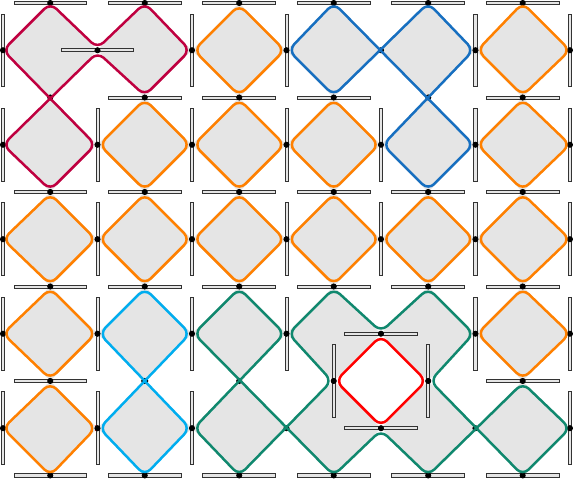}%
		\end{center}
		\caption{Mirror configuration in $\bb Z^2$ with loops.
			The  boundary condition favours loops surrounding black faces.
		}
		\label{fig mirrors}
	\end{figure}
	
	The parameters of the model are numbers 
	$p_{\r v},p_{\r h},p_{\es}\in[0,1]$ satisfying 
	$p_{\r v}+p_{\r h}+p_{\es}=1$, as well as $n>0$.  Here $n$ plays the
	same role as the spin-parameter in the quantum system, but is not
	restricted to be an integer.  A mirror configuration is chosen at
	random, with probability 
	\be
	\PP_{\Lambda,n}^{\text{mir}}(\xi)\propto 
	p_{\r v}^{\#\{x\in\Lambda:\xi_x=\r v\}}
	p_{\r h}^{\#\{x\in\Lambda:\xi_x=\r h\}}
	p_{\es}^{\#\{x\in\Lambda:\xi_x=\es\}}
	n^{\ell(\xi)}.
	\ee
	We are interested in limits of these measures as 
	$\Lambda\uparrow\bb Z^2$, in
	particular ones which are not translation-invariant.  Fix a colouring
	of the faces of $\bb Z^2$, black and white 
	in a chessboard pattern.  By deterministically placing mirrors
	in a circuit surrounding the origin, which intersect either only
	white or only black faces, one obtains measures in
	connected sets $\Lambda$ with a preference for loops around either
	black or white faces.  The larger $n$ is, the stronger this
	preference.   By this mechanism, one obtains non-translation-invariant
	(but periodic) Gibbs states.
	
	\begin{theorem}\label{thm mirror}
		In the mirror model, with $p_{\r v},p_{\r h}>0$ and $p_{\es}\in[0,1)$,
		for $n$ large enough there are two non-translation invariant, periodic
		Gibbs measures $\PP^{\bullet}_n$ and $\PP^{\circ}_n$ which are
		translations of one unit (diagonally) of each other.  Under
		$\PP_n^\bullet$, the set of black faces which are surrounded by a
		single loop forms an infinite component whose complement has only
		finite connected components (and vice versa for $\PP_n^\circ$ and
		white faces).  Correlations decay exponentially in both.  
	\end{theorem}
	
	Note that for $p_\es=0$, loops are non-crossing
	and we obtain the loop-representation of the
	critical FK-model with $q=n^2$, which is analysed in detail in
	\cite{DC-li-manolescu,DCST-fk-percolation}.  
	There the theorem above  holds for $n>2$ (i.e.\ $q>4$). 
	
	Now suppose that we rescale the lattice $\bb Z^2$ (still rotated
	$45^\circ$) by $\eps$ in the vertical direction, and at the same time
	rescale the parameters:
	$p_{\es}=u\eps$,
	$p_{\r h}=(1-u)\eps$, and 
	$p_{\r v}=1-\eps$ with $u\in[0,1)$.
	See Figure \ref{fig mirrors scaled}.
	In the limit $\eps\to0$ one obtains a continuous loop-model based on
	Poisson processes.  This is precisely the probabilistic model 
	described in Section  \ref{section:prob-rep}, with $u$ the same as in
	\eqref{eq:ham}, where we prove results analogous to 
	Theorem \ref{thm mirror}.
	(Our approach is to work directly in the continuous model.
	Another feasible approach would be to work in the discretized model
	and obtain results which are uniform in $\eps$, but there is no clear
	advantage to this method.)

	\begin{figure}[ht]
		\begin{center}
			\includegraphics[scale=1]{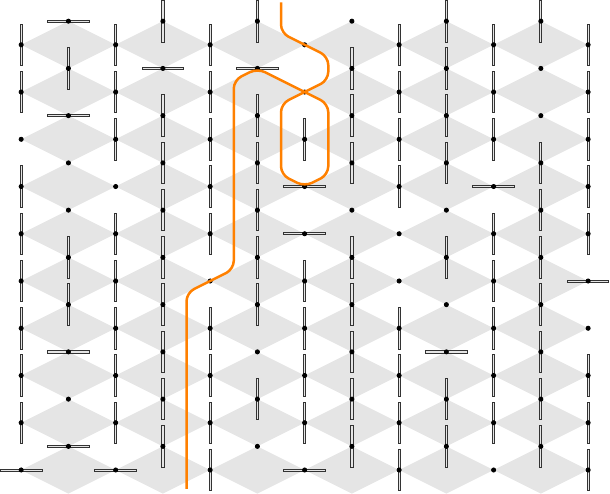}%
		\end{center}
		\caption{Rescaled mirror configuration with vertical distances scaled
			by $\eps$.  In the limit $\eps\to0$ we
			obtain the continuous loop-model \eqref{eq measure periodic}
			which is a probabilistic
			representation of the quantum spin system \eqref{eq:ham};
			horizontal mirrors become double bars $\dbar$ and missing mirrors
			become crosses $\cross$.  (Only part of one loop is drawn in this
			picture.)
		}
		\label{fig mirrors scaled}
	\end{figure}
	
	\subsection{Organization of the paper}
	In Section \ref{section:prob-rep} we define the probabilistic
	representation of the model, and state our main probabilistic results, Theorems 
	\ref{thm hole perimeter} and
	\ref{qrepair:thm:convergence}. 
	We then  show how these theorems imply Theorem
	\ref{thm:dimersation}. 
	In Section \ref{sec 1d}, we prove the loop model results,  Theorems
	\ref{thm hole perimeter} and \ref{qrepair:thm:convergence}.

	\subsection*{Acknowledgements}
	
	We are delighted to thank
	Ron Peled, Daniel Ueltschi, Bruno Nachtergaele,  Wojciech
        Samotij and Simone Warzel 
	for very useful discussions about this project, and Ulrik
        Thinggaard Hansen for a careful reading of a draft of the
        paper. 
	JEB gratefully
	acknowledges hospitality at TU Wien and at Aalto University.
	KR gratefully
	acknowledges hospitality at Chalmers / University of Gothenburg.

	\section{Probabilistic framework}\label{section:prob-rep}

	Our proofs rely on a well-known probabilistic representation 
	\cite{toth,an,ueltschi} where the quantum
	system is expressed in terms of a process of random loops.
	In Section \ref{sec 1d} 
	we work exclusively in this probabilistic framework.  The purpose of
	the present section is to describe the random loop model, to state our
	main results for the loop model, and provide `translations' of these
	results back to the quantum system.

	\subsection{Loop model}

	Recall that the quantum model is defined on a finite graph
	$\Lambda=(\c V(\Lambda),\c E(\Lambda))$.  
	Let $\beta>0$ and $u\in[0,1)$.  Let 
	$\PP_1$ denote the law of the superposition 
	of two independent Poisson
	point processes on $\c E(\Lambda)\times[0,\beta]$, 
	the first of  intensity $u$ and 
	whose points we denote by $\cross$ (called a
	\emph{cross}), and the second of  intensity $1-u$ and 
	whose
	points we denote by $\dbar$ (called a \emph{double bar}). 
	We write $\om$ for a configuration of this
	process, and the set of such configurations is
	$\Om=\Om_{\Lambda,\beta}=(\cup_{k\ge0}\c W_k)^2$, where $\c W_k$ is
	the set of subsets of $\c E(\Lambda)\times[0,\beta]$ of cardinality
	$k$.   A point of $\omega$ (which is a triplet of an edge $e\in\c E(\Lambda)$, a time $t\in[0,\beta]$ and either $\cross$ or $\dbar$) is called a \emph{link}. 
	When $\Lambda\se\bb Z^d$ we let $\PP_1$ be (the restriction of)
	a Poisson process in $\c E(\bb Z^d)\times\RR$.
	
	Each configuration $\om$ gives a set of loops, best understood by
	looking at Figure \ref{fig loops} and defined formally as follows. 
	First, we identify the endpoints $0$ and $\beta$ of $[0,\beta]$ so
	that it forms a circle;  in particular an interval
	$(a,b]\subset[0,\beta]$ with $a>b$ is defined as
	$(a,\beta]\cup(0,b]$.
	Let $\om\in\Om$ and consider the set $\c I(\om)$ of maximal
	intervals $\{x\}\times(a,b]$, $x\in \c V(\Lambda)$ and
	$(a,b]\se[0,\beta]$, which are not adjacent to a link; that is,
	intervals $\{x\}\times(a,b]$ such that there is a link of $\om$ at
	$(\{x,x'\},a)$ and at $(\{x,x''\},b)$ for some vertices 
	$x', x''\sim x$, 
	and no link at $(\{x,x'''\},t)$ for all $t\in(a,b)$ and all
	$x'''\sim x$.  
	
	Incident to a link of $\om$ at, say, $(\{x,x'\},t)$, there are four
	intervals in $\c I$: 
	$\{x\}\times(t,b_x]$, $\{x\}\times(a_x,t]$, 
	$\{x'\}\times(t,b_{x'}]$, and $\{x'\}\times(a_{x'},t]$.  We say that:
	\begin{itemize}[leftmargin=*]
		\item if the link of $\om$ at $(\{x,x'\},t)$ is a cross $\cross$, then
		$\{x\}\times(t,b_x]$ and $\{x'\}\times(a_{x'},t]$ are connected, 
		and $\{x\}\times(a_x,t]$ and $\{x'\}\times(t,b_{x'}]$ are connected;
		\item if the link of $\om$ at $(\{x,x'\},t)$ is a double-bar $\dbar$, then
		$\{x\}\times(t,b_x]$ and $\{x'\}\times(t,b_{x'}]$ are connected, 
		and $\{x\}\times(a_x,t]$ and $\{x'\}\times(a_{x'},t]$ are connected.
	\end{itemize}
	The loops of $\om$ are then the connected components of $\c I$ under
	this connectivity relation.  
	
	Write $\ell(\om)$ for the number of loops of $\om$.
	The loop measure
	we study is the Poisson point process $\PP_1$ re-weighed by
	$n^{\ell(\om)}$, 
	\phantomsection\label{not w}%
	and we denote it by $\PP^{\p}_{\Lambda,\beta,n,u}$: 
	\be\label{eq measure periodic}
	\PP^{\r{per}}_{\Lambda,\beta,n,u}[A]
	=\frac{1}{Z}\int \text{d} \PP_1(\om) \; 
	n^{\ell(\om)} \mathbbm{1}_A(\om),
	\ee
	where 
	$Z=Z_{\Lambda,\beta,n,u}=\int \text{d} \PP_1(\om)  n^{\ell(\om)}$. 
	To abbreviate we sometimes write simply $\PP_n$ for 
	$\PP^{\p}_{\Lambda,\beta,n,u}$;
	note that when $n=1$ we recover $\PP_1$.
	The superscript $^\p$ indicates that loops are counted with periodic
	boundary condition in the `vertical' direction.
	Note that the measure $\PP_n$ is well-defined also for
	non-integer $n$, although connection to the quantum system
	\eqref{eq:ham} requires $n$ to be an integer at least 2.

	\begin{figure}[th]
		\begin{center}
			\includegraphics[scale=.85]{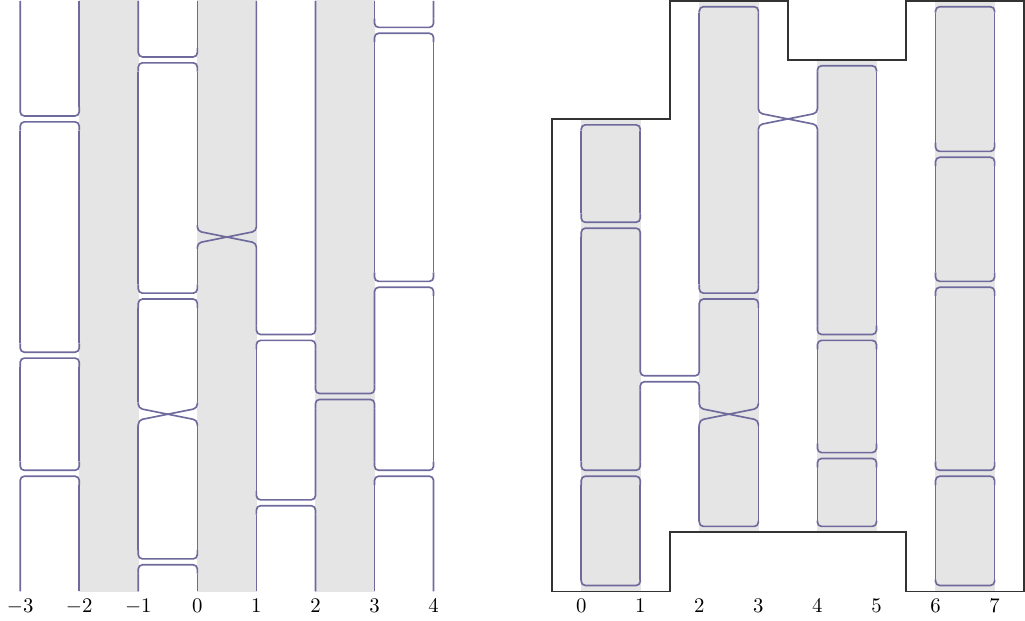}%
		\end{center}
		\caption{Two pictures of 
			configurations $\omega$ in dimension $d=1$.
			On the left, in 
			$\Lambda_L\times[0,\beta]$ (where
			$L=4$) with periodic boundary condition,
			and on the right in a primal domain $\c D_\Gamma$ with primal
			boundary condition.  
			The links ($\cross$ and
			$\dbar$) create loops, which in the left picture wrap around in the
			vertical direction, while in the right picture they are reflected on the
			boundary of the domain.  Primal columns are in both pictures drawn
			grey, dual columns white.  In the left picture, since $L=4$ is even,
			the configurations with the most loops (given a number of links) 
			have small loops stacked in dual columns, as they are 
			in the leftmost column of
			that picture.  In the right picture, small \emph{primal} loops are instead
			favoured, as in the rightmost column of that picture.  
		}
		\label{fig loops}
	\end{figure}

	To state our theorems about the loop model, we need
	additional  boundary conditions. When $\Lambda\subset\ZZ$, we 
	think of edges $e\in\c E(\Lambda)$ as elements of $(\bb Z+\tfrac12)$, where
	$e=x+\tfrac12$ connects the vertices $x$ and $x+1$.  
	Then we call $e$ \emph{primal} if its right endpoint $x+1$ is odd,
	respectively \emph{dual} if its left endpoint $x$ is odd.
	If $e$ is primal (dual) we call $\{e\}\times[0,\beta]$ a 
	\emph{primal (dual) column}, and in our pictures we always represent
	primal columns as grey, dual columns as white.
	We define a \emph{(rectangular) circuit} $\Gamma$%
	\phantomsection\label{not Gamma}
	as a simple closed curve in $\bb R^2$ made up of 
	horizontal intervals of the form $[(e-1,t),(e+1,t)]$ 
	and vertical intervals of the form $\{e\}\times[s,t]$
	for $e\in (\mathbb Z+\tfrac12)$ and $s<t$.
	A circuit $\Gamma$
	is called \emph{primal} if all its vertical intervals are in  \emph{dual}
	(white) columns  (this seemingly
	counterintuitive terminology is chosen so that primal domains favour
	primal loops, see Figure \ref{fig loops} and the discussion below).
	Conversely, $\Gamma$
	is called \emph{dual} if all its vertical intervals are in primal
	(grey) columns.

	Given a rectangular circuit $\Gamma$, the region $\c{D}_\Gamma$%
	\phantomsection\label{not domain}
	of  $\bb R^2$ enclosed by $\Gamma$ is called a \emph{domain},
	and by convention we take $\c{D}_\Gamma$ to be an open subset of 
	$\bb R^2$ (i.e.\ the boundary $\Gamma$ is not part of the domain).
	Moreover, it is called a \emph{primal domain} if $\Gamma$ is a primal
	circuit, and a  \emph{dual domain} if $\Gamma$ is a dual
	circuit.  Important examples of primal and dual domains are given by 
	\be\label{eq square domain}
	\c D_{L,\beta}=(-L+\tfrac{1}{2},L+\tfrac{1}{2})\times
	(-\tfrac{\beta}{2},\tfrac{\beta}{2}),
	\ee
	which we note is a primal
	domain for $L$ odd, and dual for $L$ even. 
	On a primal domain $\c D_\Gamma$, define 
	\emph{primal  boundary  conditions}:  when counting the number of
	loops $\ell(\omega)$,
	each point $(x,t)\in\Gamma$ satisfying $x\in 2\bb Z$ 
	is identified with $(x+1,t)\in\Gamma$ (note that these identifications
	occur on the `top and bottom' of $\c D_\Gamma$).
	To picture this, one can imagine there being a double
	bar fixed at the point $(e,t)$, where $e=x+\tfrac{1}{2}$, for each
	primal edge $e$ such that $(e,t)\in\Gamma$, see
	Figure \ref{fig loops}.
	Dual
	boundary conditions are defined similarly on a dual domain.  
	
	When $\c D_\Gamma$ is a primal domain, let $\PP_{\Gamma,n,u}^1$ denote the measure defined as in 
	\eqref{eq measure periodic}
	but with links restricted to $\c D_\Gamma$
	and with loops counted according to the primal boundary condition. 
	Similarly, write 
	$\PP_{\Gamma,n,u}^2$ for the corresponding measure defined in
	a dual domain with dual boundary condition.  To abbreviate,
	when there is no risk for confusion,  we simply write $\PP_n$.
	
	In the discussion above, we have considered only simply-connected
	domains $\c D_\Gamma$, but the extension to non-simply-connected
	domains is straightforward.
	
	Central to our analysis is what we call  
	\emph{trivial loops}, defined as a loops in
	the configuration $\omega$ which visit only two double bars $\dbar$ 
	in the same column.   If this column
	is primal (grey) we call the trivial 
	loop \emph{primal}, alternatively
	\emph{dual} if the column is dual (white).  
	Two trivial loops
	are called \emph{adjacent} if they either span the same edge and share
	a double-bar, or there are $x\in\bb Z$ and  $t\in\bb R$ such that one
	loop contains $(x,t)$ and the other $(x+1,t)$.
	Note that if two trivial loops are adjacent, then they are either both
	primal or both dual.
	Given $h>0$, a loop is called $h$-\emph{small} 
	(or just \emph{small})
	if it is trivial and has vertical height $<h$.  A trivial loop
	which is not small, i.e.\ has height $\geq h$, is called 
	\emph{tall}.
	A loop which is not small is called \emph{long}.

	\subsection{Main results for the loop model}
	
	We now state our main probabilistic results.
	If $\c D_\Gamma\se \bb R^2$ is a primal domain then
	informally, we think of any configuration $\omega$ consisting of only
	small, primal loops as a `ground-state configuration'
	for $\PP^1_{\Gamma,n,u}$. This is because such
	configurations maximize the number of loops, and hence
	the weight factor $n^{\ell(\omega)}$, subject to the primal
	boundary condition, for a given  number of links.
	The same logic applies to $\PP^{\r{per}}_{\Lambda_L,\beta,n,u}$
	with $L\in2\bb Z+1$, where the odd parity of $L$ favours small primal
	loops. 
	Correspondingly, the `ground-state configurations' for
	$\PP^2_{\Gamma,n,u}$ (with $\Gamma$ a dual circuit)
	and for $\PP^{\r{per}}_{\Lambda_L,\beta,n,u}$
	with $L\in2\bb Z$,
	consist of only small \emph{dual} loops.  
	Our first main result
	for the loop model  gives a probabilistic bound for the size of
	perturbations of such `ground-states'. 
	
	We state the result in the primal case.  
	Given $\k\geq0$, consider the union of 
	$\frac1{\k n}$-small primal loops. Let 
	$\c P(\om)=\c P_\k(\omega)$
	be the connected component of the boundary $\Gamma$ in this union (in the case of 
	$\PP^1_{\Gamma,n,u}$) or the connected component of the sides $\{-L+1,L\}\times[0,\beta]$ in this union
	(in the case of $\PP^{\r{per}}_{\Lambda_L,\beta,n,u}$
	with $L\in2\bb Z+1$).
	If $\k=0$, any trivial loop is regarded as small.
	Fix a point $x_0\in\c D$, where $\c D=\c D_\Gamma$
	or $\Lambda_L\times[0,\beta]$, and
	let $\c C(x_0)=\c C_\k(x_0,\omega)$ be the 
	connected component of 
	$\c D\sm\c P(\omega)$ which contains $x_0$.
	When $\c D=\c D_\Gamma$ the component
	$\c C(x_0)$ necessarily
	equals $\c D_\gamma$ for some random primal circuit
	$\gamma=\gamma(x_0,\omega)$.
	We define the perimeter $\per(\c C(x_0))$ as the length of
	$\gamma$ (thought of as a subset of $\RR^2$). 
	See Figure \ref{fig component} for an illustration.
	In the case when $\c D=\Lambda_L\times[0,\beta]$,
	it is possible for ${\c C(x_0)}$ to consist of two disjoint
	primal curves $\gamma_1,\gamma_2$ which wrap around the torus, in
	which case the perimeter is defined as the sum of their lengths.

	\begin{figure}[ht]
		\centering
		\includegraphics{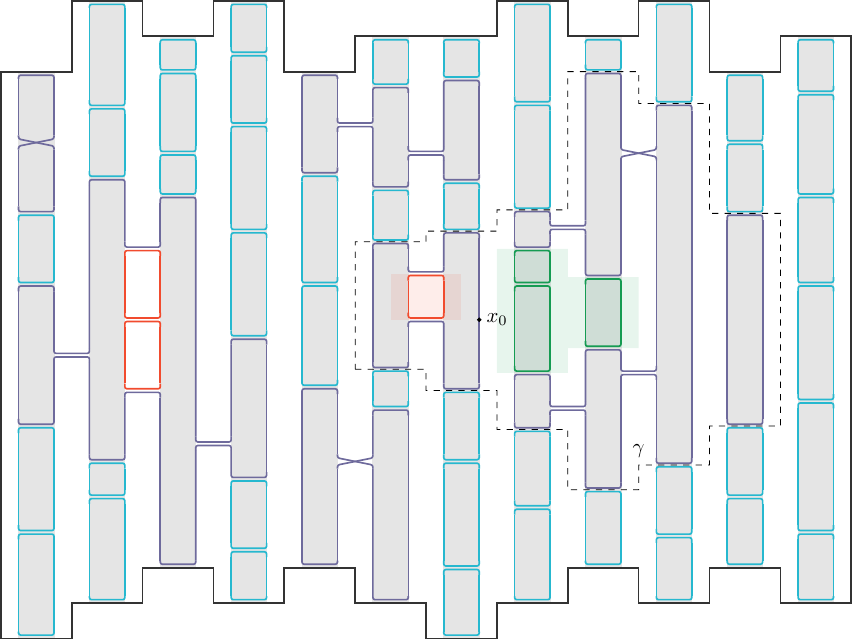}
		\caption{Illustration of the connected component $\c P(\om)$ of 
			small loops (drawn turquoise) adjacent to the boundary $\Gamma$ of a
			primal domain $\c D_\Gamma$, 
			as well as $\c C(x_0)=\c D_\gamma$
			with its boundary $\gamma$ drawn dashed.
			Long loops are drawn off-blue.
			The rightmost loop in $\c C(x_0)$ is a tall, trivial loop.
			Shaded green and orange are the  primal and dual clusters in
			$\c C(x_0)$ (as defined in Section \ref{ssec clusters}).
		}\label{fig component}
	\end{figure}

	\begin{theorem}[Perturbations of the ground state in the 
		loop model]\label{thm hole perimeter}
		For any $u\in[0,1)$ there is a constant 
		$\k_0=\k_0(u)>0$ such that for all 
		$\k\in[0,\k_0]$ there is 
		$n_0=n_0(u,\k)<\oo$
		such that the following holds.  For any $n>n_0$, 
		there is a constant $C=C(u,n,\k)>0$ such that for
		all $v>1$,
		\be
		\PP_{\Gamma,n,u}^1[\per(\c C_\k(x_0))>v] \le  \e{-Cv}
		\qquad\text{and}\qquad
		\PP^{\r{per}}_{\Lambda_L,\beta,n,u}[\per(\c C_\k(x_0))>v] 
		\le  \e{-Cv},
		\ee
		uniformly for all primal circuits $\Gamma$
		and all $L\in2\bb Z+1$ and $\beta>0$, 
		where $x_0\in\c D$ is any point in the corresponding domain
		$\c D=\c D_\Gamma$ or $\c D=\Lambda_L\times[0,\beta]$.
	\end{theorem}
	
	The corresponding result is true in the dual cases,
	i.e.\ for $\PP_{\Gamma,n,u}^2$ with $\Gamma$ a dual circuit and
	for $\PP^{\r{per}}_{\Lambda_L,\beta,n,u}$ with $L\in2\bb Z$,
	where $\c P(\om)$ should be replaced by the connected component of
	small \emph{dual} loops adjacent to the boundary or to the sides.
	\\
	
	We next consider convergence of the loop-measures (for $d=1$)
	as the lattice and/or $\beta$ tend to infinity. 
	For this we regard  $\PP_1$ as a Poisson process on the space 
	$(\ZZ+\frac12)\times\RR\times\{\cross,\dbar\}$ 
	(which we restrict to bounded domains to recover previous definitions). 
	For precise statements we need notation for point processes, 
	for which we follow \cite{daley-verejones-1,daley-verejones-2}.  Write 
	$\c M^\#$ for the set of boundedly finite measures on
	$(\ZZ+\frac12)\times\RR\times\{\cross,\dbar\}$
	(measures assigning finite mass to each compact
	set) and $\c N^\#\se\c M^\#$ for the counting-measures.
	Link-configurations $\omega$ and elements of 
	$\c N^\#$ are identified in the natural way, in particular this gives
	a definition of infinite-volume link configurations $\omega$.
	The set $\c M^\#$ is given the 
	$w^\#$-topology (see \cite[Section A2.6]{daley-verejones-1})
	and we write $\c B(\c M^\#)$ for the Borel
	$\sigma$-algebra generated by this topology.
	
	For a sequence of domains 
	$\c D_k\se\bb Z\times\bb R$, $k\geq1$,
	we write $\c D_k\nearrow \bb Z\times\bb R$
	if $\c D_k\se \c D_{k+1}$ for all $k\geq1$ and 
	$\bigcup_{k\geq1} \c D_k=\bb Z\times\bb R$.
	We also consider doubly-indexed sequences 
	$\c D_{L,\beta}$ in which case we require 
	$\c D_{L,\beta}\se \c D_{L',\beta'}$ whenever 
	$L'\geq L$ and $\beta'\geq\beta$ as well as 
	$\bigcup_{L,\beta} \c D_{L,\beta}=\bb Z\times\bb R$.
	We henceforth replace $[0,\beta]$ with $[-\beta/2,\beta/2]$
	so that 
	$\Lambda_L\times[-\beta/2,\beta/2]\nearrow\ZZ\times\RR$
	as $L,\beta\to\oo$.
	For a link-configuration $\om$  and $\alpha=1,2$, write
	\phantomsection\label{not Ealpha}%
	$\c E^\alpha=\c E^\alpha_\kappa$ 
	for the complement of the union of all unbounded
	connected
	components of small primal (if $\alpha=1$, resp.\ dual if $\alpha=2$)
	loops.  
	
	\begin{theorem}[Dimerization in the loop model for $d=1$]
		\label{qrepair:thm:convergence}
		Let $u\in[0,1)$ and $\alpha\in\{1,2\}$.  
		There is a constant  $\k_0=\k_0(u)>0$ such that for all 
		$\k\in[0,\k_0]$ there is 
		$n_0=n_0(u,\k)<\oo$
		such that the following holds.  
		Let $\c D^\alpha$ be elements of a sequence or doubly-indexed
		sequence  of primal (for $\alpha=1$)
		respectively dual (for $\alpha=2$) domains
		with $\c D^\alpha\nearrow \ZZ\times\RR$, and let
		$\Lambda_L=\{-L+1,\dots,L\}\subset\ZZ$. 
		There are two distinct infinite-volume Gibbs  measures
		$\PP^\alpha_{n,u}$, $\alpha=1,2$, such that
		as $\c D^\alpha\nearrow \ZZ\times\RR$ respectively as
		$\beta\to\infty$ and $L\to\oo$
		with $L\in2\mathbb Z+\alpha$,
		either simultaneously or in the order $\beta$ followed by $L$,
		\be
		\PP^\alpha_{\c D^\alpha,n,u} 
		\to
		\PP^\alpha_{n,u},
		\quad\text{respectively}\quad
		\PP^\p_{\Lambda_L,\beta,n,u} 
		\to
		\PP^\alpha_{n,u},
		\qquad
		\text{ weakly in } (\c M^\#,\c B(\c M^\#)).
		\ee
		The measures
		$\PP^\alpha_{n,u}$, $\alpha=1,2$, are
		supported on configurations with no
		infinite cluster of $\c E^\alpha$, they
		are $2\ZZ\times\RR$-invariant and ergodic,
		and they satisfy $\tau_{(1,0)}\PP^\alpha_{n,u}=\PP^{3-\alpha}_{n,u}$, where
		$\tau_{(1,0)}$ is the shift by $(1,0)$. 
	\end{theorem}
	
	Here, being a Gibbs measure means that the conditional distribution
	inside a bounded domain $\c D$, \emph{given} the configuration outside 
	$\c D$, is given by the expression \eqref{eq measure periodic}
	where loops inside $\c D$ are counted with respect to
	the connectivity imposed by the configuration outside $\c D$
	(see Section \ref{ssec conv} for a more detailed definition).
	Note that  Theorem \ref{thm hole perimeter} 
	(Perturbations of the ground state)
	extends to the limiting measures 
	$\PP^\alpha_{n,u}$, $\alpha=1,2$ since the bounds in that Theorem are
	uniform in the domain.
	Also note that Theorems \ref{thm hole perimeter} and
	\ref{qrepair:thm:convergence} hold for non-integer $n$.
	
	\begin{remark}
		For $\PP^\p_{\Lambda_L,\beta,n,u}$ we exclude the case $L\to\oo$
		followed by $\beta\to\oo$ in the above theorem and
		we do not expect to obtain either of the measures $\PP^\alpha_{n,u}$
		in the limit.  Intuitively, this is because if $L\to\oo$ first, then
		the boundary effect obtained by requiring $L$ to be odd is lost.  
		Note that in this case the `defect' 
		$\c C(x_0)$ in Theorem \ref{thm hole perimeter} can have arbitrarily
		large volume and bounded perimeter.
	\end{remark}

	\subsection{Translation of loop model results to the quantum system}
	
	The loop model is a powerful representation for the quantum spin
	system because the expectation of local observables of the quantum
	system can be written as probabilities of events in the loop
	model.
	Specifically, the Gibbs state $\langle\cdot\rangle_{\Lambda}$
	\eqref{eq gibbs state} is
	represented using the periodic loop-measure
	$\PP_{\Lambda,\beta,n,u}^\p$, while the seeded state
	$\langle\cdot\rangle_{\Lambda_L}^{\Psi_L}$
	\eqref{eq seeded state} is represented using
	$\PP_{{\c D}_{L,\beta},n,u}^\a$ where
	$\c D_{L,\beta}=(-L+\tfrac{1}{2},L+\tfrac{1}{2})\times
	(-\tfrac{\beta}{2},\tfrac{\beta}{2})$
	with $L\in2\bb Z+\a$
	as in \eqref{eq square domain}.
	While this is well known, we write and prove the following
	lemma which is slightly more general than results that usually 
	appear in the literature.  
	
	

	Consider a link-configuration $\om$ sampled from 
	$\PP_{\Lambda,\beta,n,u}^\p$ or
	$\PP^\alpha_{\Gamma,n,u}$, and a finite subset 
	$X=\{(x_1,t_1),\dots,(x_k,t_k)\}$ of
	$\Lambda\times[0,\beta]$ or of a domain $D_{\Gamma}$ in 
	$\bb Z\times\bb R$, respectively.
	The loops of $\om$ naturally
	produce a pairing $\ul{\pi}_X(\om)$ of the points
	$X^\pm=\{(x_i,t_i^-),(x_i,t_i^+) : i=1,\dots,k \}$.  For 
	$\ul{i}^-,\ul{i}^+\in[n]^k$ and a fixed pairing $\pi$ of $X^\pm$, we 
	write $\pi\sim(\ul{i}^-,\ul{i}^+)$  if for all pairs
	$(x_r,t_r^{\alpha_r}), (x_s,t_s^{\alpha_s})$ in $\pi$, we have
	$i_r^{\alpha_r}=i_s^{\alpha_s}$, where
	$\alpha_r,\alpha_s\in\{+,-\}$. 
	In this case we say that the pairing $\pi$ and 
	$\ul{i}^-,\ul{i}^+$ are \emph{compatible}. 
	Finally, write $\ell(\ul{\pi}_X)$ for the
	number of loops of $\om$ passing through points of $X$
	(this is a function of $\omega$ and $X$ which only depends on 
	the induced pairing $\ul{\pi}_X(\omega)$).
	
         For any finite set $\Delta$
	we write $\ket{\ul{i}}$ for elements of the usual product
	basis  of $(\bb C^n)^{\otimes \Delta}$, and 
	$E_{\ul{i}^-,\ul{i}^+}=|{\ul{i}^-}\rangle\langle {\ul{i}^+}|$ for
	the elementary operator.    Recall that
		$A(t)=\e{-t H_\Lambda}A\e{tH_\Lambda}$
	and, 
	in the case $d=1$, recall that $\Lambda_L=\{-L+1,\dots,L\}$. 
	
		\begin{lemma}\label{lem:eltry-matrices}
			Let $\Lambda$ be any finite graph, $\beta>0$. Let 
			$\Gamma=\{x_1,\dots,x_k\}, \Delta=\{y_1,\dots,y_m\}\subset\Lambda$
			and  $\ul{i}^-,\ul{i}^+\in[n]^k$, $\ul{j}^-,\ul{j}^+\in[n]^m$. 
			Write $X=(\Gamma\times\{0\})\cup(\Delta\times\{t\})$.
			We have that
			\be\label{eq:elmat-1}
			\langle E_{\ul{i}^-,\ul{i}^+}  E_{\ul{j}^-,\ul{j}^+}(t)
			\rangle_{\Lambda,\beta,n,u}
			=
			\EE^\p_{\Lambda,\beta,n,u}[n^{-\ell(\ul{\pi}_{X}(\om))} 
			\Ind{\ul{\pi}_{X}(\om) \sim (\ul{i}^-,\ul{i}^+)} ]
			\ee
			and
			\be\label{eq:elmat-2}
			\langle   E_{\ul{i}^-,\ul{i}^+} E_{\ul{j}^-,\ul{j}^+}(t)  
			\rangle_{\Lambda_L,\beta,n,u}^{\Psi_L}
			=
			\EE^\alpha_{\c D_{L,\beta},n,u}[n^{-\ell(\ul{\pi}_{X}(\om))} 
			\Ind{\ul{\pi}_{X}(\om) \sim (\ul{i}^-,\ul{i}^+)} ],
			\ee
			where $\alpha=1$ if $L$ odd, and $\alpha=2$ if $L$ even, 
			and $\Psi_L$ is given by \eqref{eq:Psi}.
		\end{lemma}
		
		\begin{proof}
			The proofs for the two cases 
			\eqref{eq:elmat-1} and \eqref{eq:elmat-2} 
			are the same up to minor notational
			changes, so we prove only \eqref{eq:elmat-1}
			and write
			$\langle\cdot\rangle$ for $\langle\cdot \rangle_{\Lambda,\beta,n,u}$
			and 
			$\EE_n$ for $\EE^\p_{\Lambda,\beta,n,u}$.
			
			Write $\Sigma$ for the set of (c\`adl\`ag) functions
			$\sigma:\Lambda\times[0,\beta]\to\{1,\dots,n\}$ (where $0=\beta$)
			that have only
			finitely many discontinuities in $t$.  For $\sigma\in\Sigma$,
			$\om\in\Om$, we write $\sigma\sim\om$ ($\sigma$ is compatible with
			$\om$) if $\sigma$ constant on each loop of $\om$. By uniformly
			colouring the loops of $\om$ with a number in $[n]$, one obtains a
			measure on $\Om\times\Sigma$ given by $\PP_1$ times the counting
			measure on compatible configurations $\sigma$.  
			As in \cite[Theorem~3.2]{ueltschi} we obtain
			\be\label{eq:pf-sigma}
			Z=\Tr\;\e{-\beta H_\Lambda}=\int_\Om \text{d}\PP_1(\om)
			\sum_{\sigma\in\Sigma}\Ind{\sigma\sim \om}=
			\int_\Om \text{d}\PP_1(\om) \;  n^{\ell(\omega)}.
			\ee
			Inserting the operator 
			$E_{\ul{i}^-,\ul{i}^+}E_{\ul{j}^-,\ul{j}^+}(t)$ 
			modifies \eqref{eq:pf-sigma}
			by introducing discontinuities in $\sigma$ at 
			the points of $X$. To
			be precise, for $\sigma\in\Sigma$,
			$\om\in\Om$, we write 
			$\sigma\sim (\om; X;\ul{i}^-,\ul{i}^+; \ul{j}^-,\ul{j}^+)$ if $\sigma$
			is constant on the loops of $\om$ except possibly at points of $X$,
			where $\sigma(x_r,0^{\alpha}) = i_r^{\alpha}$
			and $\sigma(y_s,t^{\alpha}) = j_r^{\alpha}$, for all $r=1,\dots,k$,
			$s=1,\dots,m$,
			and $\alpha\in\{+,-\}$.  
			We then have 
			\be
			\begin{split}
				\langle  E_{\ul{i}^-,\ul{i}^+} E_{\ul{j}^-,\ul{j}^+}(t)  \rangle
				&=
				\frac{1}{Z}\int_\Om \text{d}\PP_1(\om) 
				\Ind{\ul{\pi}_{X}(\om)\sim(\ul{i}^-,\ul{i}^+,\ul{j}^-,\ul{j}^+)}
				\sum_{\sigma\in\Sigma}
				\Ind{\sigma\sim (\om; X;\ul{i}^-,\ul{i}^+; \ul{j}^-,\ul{j}^+)}\\
				&=
				\frac{1}{Z}\int_\Om \text{d}\PP_1(\om) 
				\Ind{\ul{\pi}_{X}(\om)\sim(\ul{i}^-,\ul{i}^+, \ul{j}^-,\ul{j}^+)}
				n^{\ell(\om)} n^{-\ell(\ul{\pi}_{X}(\om))}\\
				&=
				\EE_n[n^{-\ell(\ul{\pi}_{X}(\om))}
				\mathbbm{1}\{\ul{\pi}_{X}(\om) \sim (\ul{i}^-,\ul{i}^+)\} ].
			\end{split}
			\ee
			The second equality holds because if 
			$\sigma\sim (\om; \ul{i}^-,\ul{i}^+;\ul{j}^-,\ul{j}^+)$
			then the loops of $\omega$ that pair points of $X^\pm$ have exactly
			one possible colouring (determined by $\ul i^\pm, \ul j^\pm$),
			whereas the remaining loops have $n$ possible colourings. 
		\end{proof}
	
	We now prove Theorem \ref{thm:dimersation}. The second part (exponential decay of correlations) was proved for $u=0$ in \cite{an, ADCW}, and there uses the FKG properties of the quantum FK percolation. We have no such properties for $u>0$. Our proof uses the full force of our work on the probabilistic representation in Section \ref{sec 1d}, making key use of
	the following estimate on the total-variation distance between the
	marginals of two measures defined in different large domains
	(proved in Section \ref{ssec conv}).  We state it for the primal case
	$\alpha=1$, although corresponding result also holds for the dual case
	$\alpha=2$.  
	
	\begin{lemma}\label{lem TV}
		Let $u\in[0,1)$.  There is an $n_0(u)$ such that for $n>n_0$ 
		there is a constant $C>0$ such that the following holds.
		Let $\c D_1$ and $\c D_2$ each be either of the form
		$\c D_\Gamma$ for a primal circuit $\Gamma$, or
		of the form
		$\Lambda_L\times[-\beta/2,\beta/2]$ for $L\in2\bb Z+1$
                and write $\EE^1_{\c D_i}$ for expectation under the measure 
                $\PP^1_{\Gamma,u,n}$ in the former case, respectively 
                $\PP^\p_{\Lambda_L,\beta,u,n}$ in the latter case.
                Let $\c A$ be a domain and $\c B$ a primal domain
                satisfying
                $\c A\se\c B\se \c D_1\cap\c D_2$.
                Let  $X=X(\om)$ be a random variable 
                depending only on the loops intersecting  
                $\c A$ and satisfying $|X|\leq K$ almost surely.
                Then we have
                \be
                |\EE^1_{\c D_1}[X]-\EE^1_{\c D_2}[X]|
                \leq 2K\vol(\c A) \e{-Cd(\c A,\c B^c)},
                \ee
                where $d(\c A,\c B^c)$ denotes the minimal distance between 
                points in $\c A$ and outside $\c B$, and
                the volume $\vol(\c A)$
                is the total (vertical) length of the intervals of
                $\c A\cap(\bb Z\times \RR)$
	\end{lemma}
	
	We also need \emph{Mecke's formula},
	which provides a method for conditioning on the exact locations of
	points in a Poisson process (see e.g.\
	\cite[Theorem 4.4]{last-penrose}).
	Recall that $\PP_1$ is the law of a Poisson process on 
	$(\ZZ+\frac12)\times\RR\times\{\cross,\dbar\}$.
	We write $\mu$ for its intensity measure, which is a product of $u$
	times Lebesgue measure on $(\ZZ+\frac12)\times\RR$
	with $1-u$ times the same.
	
	\begin{lemma}[Mecke's formula]
		\label{lem mecke}
		For $f$   any bounded measurable function of pairs of locally
		finite subsets of 
		$(\ZZ+\frac12)\times\RR\times\{\cross,\dbar\}$,
		\be
		\EE_1\Big[
		\sum_{\eta\se\omega, |\eta|=m} f(\eta,\omega)
		\Big]
		=\int \dd \mu^{\odot m}(\eta)\; \EE_1[f(\eta,\omega\cup\eta)]
		\ee
		where 
		\be
		\dd \mu^{\odot m}(\{x_1,\dotsc,x_m\})
		=\frac1{m!} \sum_{\pi\in S_m}
		\dd\mu^{\otimes m}(x_{\pi(1)},\dotsc,x_{\pi(m)})
		\ee
		is the symmetrized $m$-fold product measure.
	\end{lemma}
	
\begin{proof}[Proof of Theorem \ref{thm:dimersation}]		
Let $A$, $B$ be local operators with respective supports $\Lambda(A)$,
$\Lambda(B)$.  We write  $A$ as a
linear combination of elementary matrices:
\be\label{eq:mainpf-lincomb}
A=\sum_{\ul{i}^-,\ul{i}^+\in [n]^{\Lambda(A)}}
A_{\ul{i}^-,\ul{i}^+}
E_{\ul{i}^-,\ul{i}^+}\otimes\one_{\Lambda\sm\Lambda(A)}.
\ee
Thus, by Lemma \ref{lem:eltry-matrices},
\be\label{eq state prob}
\langle A\rangle 
=
\sum_{\ul{i}^-,\ul{i}^+} A_{\ul{i}^-,\ul{i}^+}
\langle E_{\ul{i}^-,\ul{i}^+} \rangle
=
\sum_{\ul{i}^-,\ul{i}^+} A_{\ul{i}^-,\ul{i}^+}
\EE_{\Lambda_L,\beta,n,u}[n^{-\ell(\ul{\pi}_{\Lambda(A)\times\{0\}}(\om))}
\mathbbm{1}\{\ul{\pi}_{\Lambda(A)\times\{0\}}(\om) \sim (\ul{i}^-,\ul{i}^+)\} ].
\ee
Here $\langle\cdot\rangle$ denotes either 
$\langle\cdot\rangle_{\Lambda_L,\beta,n,u}$ or 
$\langle\cdot\rangle_{\Lambda_L,\beta,n,u}^{\Psi_L}$, and
$\EE_{\Lambda_L,\beta,n,u}$ denotes $\EE_{\Lambda_L,\beta,n,u}^\p$ in the former case,
respectively $\EE_{\c D_{L,\beta},n,u}^\alpha$ in the latter case,
where we recall the domain $\c D_{L,\beta}$ 
from \eqref{eq square domain}.

We show that the right-hand-side
of \eqref{eq state prob}
converges as $L,\beta\to\infty$ with $L\in2\ZZ+\alpha$:
\be\label{qrepair:eq:loop-observable-converges}
\begin{split}
  &\EE_{\Lambda_L,\beta,n,u}[n^{-\ell(\ul{\pi}_{\Lambda(A)\times\{0\}}(\om))}
  \mathbbm{1}\{\ul{\pi}_{\Lambda(A)\times\{0\}}(\om) \sim (\ul{i}^-,\ul{i}^+)\} ]\\
  &\qquad \to 
  \EE^{\alpha}_{n,u}[n^{-\ell(\ul{\pi}_{\Lambda(A)\times\{0\}}(\om))}
  \mathbbm{1}\{\ul{\pi}_{\Lambda(A)\times\{0\}}(\om) \sim (\ul{i}^-,\ul{i}^+)\} ].
\end{split}
\ee
Writing $f=n^{-\ell(\ul{\pi}_{\Lambda(A)\times\{0\}}(\om))}
\mathbbm{1}\{\ul{\pi}_{\Lambda(A)\times\{0\}}(\om) \sim (\ul{i}^-,\ul{i}^+)\}$,
for $N\in\NN$ let $f_N=f\cdot\mathbbm{1}(\c{A}_N)$, where $\c{A}_N$ 
is the event that all loops passing through $\Lambda(A)\times\{0\}$ 
stay within the box $[-N,N]^2\subset\ZZ\times\RR$. From Theorem
\ref{qrepair:thm:convergence}, we have that for all $\eps>0$, there  
exists $N\in\NN$ such that for all $L,\beta>2N$,
$\EE_{\Lambda_L,\beta,n,u}[A_N]>1-\eps$ and also
$\EE_{n,u}[A_N]>1-\eps$.  
We then have, since $f<1$,
\bes
\left|\EE_{\Lambda_L,\beta,n,u}[f] - \EE_{\Lambda_L,\beta,n,u}[f_N]\right|
\le 
\EE_{\Lambda_L,\beta,n,u}[A_N^c]
\le \eps, 
\ees
and 
\bes
\left|\EE_{n,u}[f] - \EE_{n,u}[f_N]\right|
\le 
\EE_{n,u}[A_N^c]
\le \eps, 
\ees
and finally by Theorem \ref{qrepair:thm:convergence} we know since $f_N$ is local that
$\EE_{\Lambda_L,\beta,n,u}[f_N]\to\EE_{n,u}[f_N]$. Hence for $L,\beta$
large enough, we have
$|\EE_{\Lambda_L,\beta,n,u}[f]-\EE_{n,u}[f]|<3\eps$, and so
\eqref{qrepair:eq:loop-observable-converges} holds, and so \eqref{eq
  state prob} converges. Note that the limits can be taken in any
order or together in the case of  
$\EE_{\c D_{L,\beta},n,u}^\alpha$, respectively $\beta$ and $L$ together or first
$\beta\to\oo$ followed by $L\to\oo$ in the case of 
$\EE_{\Lambda_L,\beta,n,u}^\p$.
This gives the convergence of the finite volume states,
and that the limiting states are different along the limits $L$ even
and $L$ odd.\\

We turn to the second part of Theorem \ref{thm:dimersation},
exponential decay of truncated correlations.
We work with $\alpha=1$ (the proof for $\alpha=2$ is the same),
$\beta>t$, and $L$ large enough that
$\Lambda_L\supseteq\Lambda(A)\cup\Lambda(B)$.
To emphasize the dependence on the domain, in this part we write
$\c D$ for $\Lambda_L\times(-\beta/2,\beta/2)$ and
$\PP_{\c D}$ for the loop-measure in this domain
(which is either $\EE_{\c D_{L,\beta},n,u}^\alpha$ or
$\EE_{\Lambda_L,\beta,n,u}^\p$ but the proof works the same for both
cases). 
Observe that:
\be\label{looptoquantum:eq:proof-working0}
\begin{split}
  \langle A;B(t)\rangle
  &=
  \sum_{\ul{i}^-,\ul{i}^+,\ul{j}^-,\ul{j}^+} 
  A_{\ul{i}^-,\ul{i}^+} B_{\ul{j}^-,\ul{j}^+}
  \langle E_{\ul{i}^-,\ul{i}^+} ;
  E_{\ul{j}^-,\ul{j}^+}(t) \rangle,
\end{split}
\ee
where the indices $\ul i^\pm$ and $\ul j^\pm$
belong to $[n]^{\Lambda(A)}$ and $[n]^{\Lambda(B)}$ respectively.  
We define
\be
\begin{split}
  X_A(\omega)   &=
  n^{-\ell(\ul{\pi}_{\Lambda(A)\times\{0\}}(\om))} 
  \mathbbm{1}\{\ul{\pi}_{\Lambda(A)\times\{0\}}(\om) \sim (\ul{i}^-;\ul{i}^+)\}, \\
  X_B(\omega)  &=
  n^{-\ell(\ul{\pi}_{\Lambda(B)\times\{t\}}(\om))} 
  \mathbbm{1}\{\ul{\pi}_{\Lambda(B)\times\{t\}}(\om) \sim (\ul{j}^-;\ul{j}^+)\},\\
  X_{AB}(\omega)
  &=
  n^{-\ell(\ul{\pi}_{(\Lambda(A)\times\{0\})\cup(\Lambda(B)\times\{t\})}(\om))} 
  \mathbbm{1}\{\ul{\pi}_{(\Lambda(A)\times\{0\})
    \cup(\Lambda(B)\times\{t\})}(\om)
  \sim (\ul{i}^-,\ul{j}^-;\ul{i}^+,\ul{j}^+)\}.
\end{split}
\ee
By  Lemma \ref{lem:eltry-matrices} we have
\be
\begin{split}
  \langle E_{\ul{i}^-,\ul{i}^+} ;
  E_{\ul{j}^-,\ul{j}^+}(t) \rangle
  &=
  \langle E_{\ul{i}^-,\ul{i}^+} 
  E_{\ul{j}^-,\ul{j}^+}(t) \rangle
  -
  \langle E_{\ul{i}^-,\ul{i}^+} \rangle
  \langle E_{\ul{j}^-,\ul{j}^+}(t) \rangle\\
  &=\EE_{\c D}[X_{AB} ]- \EE_{\c D}[X_A ]\EE_{\c D}[X_B ].
\end{split}
\ee
Informally, Theorem \ref{thm hole perimeter}
shows that when $\Lambda(A)\times\{0\}$ and $\Lambda(B)\times\{t\}$ are
far apart, then with high probability the set $\Lambda(A)\times\{0\}$ is
surrounded by a circuit of small primal loops which separates it from
$\Lambda(B)\times\{t\}$.  Since no loops can cross this circuit,
it follows firstly  that $X_{AB}$ factorizes as $X_AX_B$, and secondly
the two factors $X_A$ and $X_B$ are conditionally independent given
the circuit.  We now make this rigorous, relying on 
Lemma \ref{lem TV}.
See Figure \ref{fig circuit} for an illustration.

  Write $R$ for the minimal distance between
  $\Lambda(A)\times\{0\}$ and
  $\Lambda(B)\times\{t\}$, with respect to the $1$-distance 
  in $\bb R^2$;  thus $R= d(\Lambda(A),\Lambda(B))+|t|$.  
Let $\c D_A\se \ZZ\times\RR$
be the set of points within distance $R/2$ of
$\Lambda(A)\times\{0\}$.  Also let
$\c D_B\se \ZZ\times\RR$ be a primal domain
containing $\Lambda(B)\times\{t\}$ whose boundary is at
distance $R/4$ from $\Lambda(B)\times\{t\}$
(up to an additive constant).  Thus $d(\c D_B,\c D_A)\geq R/4$.
Let $U$ be the event that there is a closed circuit 
of small primal loops contained 
in $\c D_A$ which surrounds $\Lambda(A)\times\{0\}$.
Let $V$ be the event that all loops containing points of
$\Lambda(B)\times\{t\}$ are contained in $\c D_B$.
We have $\PP_{\c D}(U^c)\leq|\Lambda(A)|\e{-CR/4}$
and
$\PP_{\c D}(V^c)\leq|\Lambda(B)|\e{-CR/4}$ 
for some $C>0$, by
Theorem \ref{thm hole perimeter} and a union bound 
(where we fix $\k\in(0,\k_0(u)]$ such that
the maximal height $\frac1{\k n}$ of a tall loop is $\ll R$; this is
only really relevant if one of $\c D_A,\c D_B$ is `above' the other).
On $U$ we have $X_{AB}=X_AX_B$ so
\be \label{looptoquantum:eq:proof-working3}
\EE_{\c D}[X_{AB} ]=
\EE_{\c D}[X_{AB} \one_{U^c\cup V^c}]+
\EE_{\c D}[X_{A}Y_B\one_U ]
\ee
where $Y_B=X_B\one_V$.  Note that $Y_B$ depends only on the
configuration of links in
${\c D_B}$, and that
$|\EE_{\c D}[X_{AB} \one_{U^c\cup V^c}]|\leq
\PP_{\c D}(U^c\cup V^c)\leq(|\Lambda(A)|+|\Lambda(B)|)\e{-CR/4}$ 
since $X_{AB}$ is bounded above by 1.
We will now show that $\EE_{\c D}[X_{A}Y_B\one_U ]$
is closely approximated by
$\EE_{\c D}[X_{A} ]\EE_{\c D}[X_B]$.

\begin{figure}[bht]
  \centering
  \includegraphics{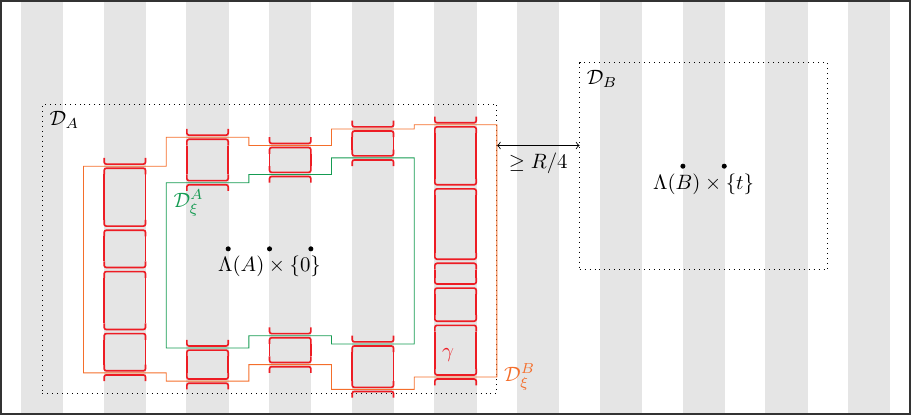}
  \caption{Illustration of the event $U$.  In red, $\gamma$ is the
    innermost circuit of small loops surrounding
    $\Lambda(A)\times\{0\}$, and the links (all $\dbar$) included in
    $\gamma$ form the configuration $\xi$.  The inside and outside of
    $\gamma$ define two primal domains 
    $\c D^A_\xi$ and $\c D^B_\xi$.  Since $\c D_B$ is at distance at least
    $R/4$ from $(\c D_\xi^B)^c$,  Lemma \ref{lem TV} implies that the
    measures defined in $\c D_\xi^B$ and in the whole domain $\c D$
    have very close marginal distributions in $\c D_B$.
  }\label{fig circuit}
\end{figure}
		
On $U$, let $\gamma$ be the innermost 
circuit of small loops surrounding 
$\Lambda(A)\times\{0\}$
(i.e.\ closest to $\Lambda(A)\times\{0\}$).  Let
$\Xi\subset\om$ be the links on $\gamma$.  Note that $\Xi(\om)$
consists of only double bars.
Using Mecke's formula (Lemma \ref{lem mecke})
\be\begin{split}\label{looptoquantum:eq:proof-working1}
  \EE^1_{\c D}[X_{A}Y_B\one_U ]&=
  \frac1{Z_{\c D}}\sum_{r\geq0}
  \EE_1\Big[
  \sum_{\substack{\xi\se\omega\cap\c D_A \\ |\xi|=r}}
  X_A(\omega) Y_B(\omega) \one\{\Xi(\omega)=\xi\}
  n^{\ell(\omega)}\Big]\\
  &=\frac1{Z_{\c D}}\sum_{r\geq0}
  \int\dd\mu^{\odot r}(\xi)
  \EE_1[X_A(\omega\cup\xi)Y_B(\omega\cup\xi)
  n^{\ell(\omega\cup\xi)}
  \one\{\Xi(\omega\cup\xi)=\xi\}].
\end{split}\ee
On the event $\Xi(\omega\cup\xi)=\xi$, we may split
$\c D=\gamma\cup\c D^A_\xi\cup \c D^B_\xi$, where 
$\c D^A_\xi$ is the domain enclosed by $\gamma$ and
$\c D^B_\xi=\c D\sm (\c D^A_\xi\cup\gamma)$.  
Both $\c D^A_\xi$and $\c D^B_\xi$ are primal domains and
they are separated by $\xi$ ($\c D^B_\xi$ need not be  simply
connected but this makes no difference).  
We may write $\omega=\omega_A\cup\omega_B$,
where $\omega_A,\omega_B$ are the links in 
$\c D^A_\xi, \c D^B_\xi$, respectively.  
Since no loops can traverse
between the two domains, we have 
$\ell(\omega\cup\xi)=\ell(\omega_A)+\ell(\omega_B)+\ell(\gamma)$,
where $\ell(\gamma)=|\xi|/2$ is the number of loops constituting
$\gamma$ (each loop on $\gamma$ is bounded by two double-bars of
$\xi$). 
The configuration $\om_A$ is constrained to belong to the event
$W_\xi$ that no small loops are adjacent to the boundary 
$\partial\c D^A_\xi$, but  $\om_B$ has no such constraint. 
Thus
\be\begin{split}\label{looptoquantum:eq:proof-working2}
  &\EE_1[X_A(\omega\cup\xi)Y_B(\omega\cup\xi)
  n^{\ell(\omega\cup\xi)}
  \one\{\Xi(\omega\cup\xi)=\xi\}]\\
  &=n^{r/2}
  \EE_1[X_A(\omega_A)
  n^{\ell(\omega_A)}\one_{W_\xi}(\om_A)]
  \EE_1[Y_B(\omega_B)
  n^{\ell(\omega_B)}]\\
  &=n^{r/2}
  \EE_1[X_A(\omega_A)
  n^{\ell(\omega_A)}\one_{W_\xi}(\om_A)]
  Z_{\c D_\xi^B}
  \EE^1_{\c D_\xi^B}[Y_B].
\end{split}\ee
Consider the last factor, $\EE^1_{\c D_\xi^B}[Y_B]$.  
We apply Lemma \ref{lem TV}
with $\c D_1=\c D$, $\c D_2=\c B=\c D_\xi^B$
and $\c A=\c D_B$ and use the fact that 
$\c B^c=(\c D_\xi^B)^c$ is at distance at least $R/4$
from $\c A=\c D_B$, to conclude that
\be
|\EE^1_{\c D_\xi^B}[Y_B]-
\EE_{\c D}[Y_B]|\leq 2\vol(\c D_B) \e{-CR/4}
\ee
for some $C>0$.
Putting this back into \eqref{looptoquantum:eq:proof-working2},
and reversing the steps in \eqref{looptoquantum:eq:proof-working1}, 
we conclude that
\be
|\EE_{\c D}[X_{A}Y_B\one_U ]-
\EE_{\c D}[X_{A}\one_U]\EE_{\c D}[Y_B]|
\leq  2\vol(\c D_B)\e{-CR/4}.
\ee
Finally we use that
$|\EE_{\c D}[Y_B]-\EE_{\c D}[X_B]|\leq \PP_{\c D}(V^c)
\leq |\Lambda(B)|\e{-CR/4}$
and
$|\EE_{\c D}[X_A\one_U]-\EE_{\c D}[X_A]|\leq \PP_{\c D}(U^c)
\leq |\Lambda(A)|\e{-CR/2}$, as well 
as \eqref{looptoquantum:eq:proof-working3},
to conclude that
\be
|\EE_{\c D}[X_{AB}]-\EE_{\c D}[X_{A}]\EE_{\c D}[X_{B}]|\leq
c_0 |\Lambda(A)|\vol(\c D_B) \e{-c_1R}
\ee
for some $c_0,c_1>0$.
Since $\vol(\c D_B) \leq |\Lambda(B)|(R/4)^2$ and 
$R= d(\Lambda(A),\Lambda(B))+|t|$, the result follows after suitably
adjusting $c_0,c_1$.
(We use that, in \eqref{looptoquantum:eq:proof-working0}, 
the maximum absolute value
of the matrix entries $A_{\ul{i}^-,\ul{i}^+}$ is bounded above by
the matrix norm $\|A\|$, and similarly for $B$.)

  Finally we prove the third part, that the states
  $\langle\cdot\rangle_{\alpha}$ are gapped.  
  Again we focus on the case $\alpha=1$ and thus let 
  $L\in2\bb Z+1$.
  Let $\beta'>\beta$ and 
  let $\c A=\Lambda_L\times[-\frac\beta8,\frac\beta8]$
  and $\c B=\Lambda_L\times[-\frac\beta4,\frac\beta4]$.
  Using Lemma \ref{lem:eltry-matrices}, specifically 
  \eqref{eq state prob} with $A=H_{\Lambda_L}$, as well as 
  Lemma \ref{lem TV} with $\c D_1=\c D_{L,\beta}$
  and $\c D_2=\c D_{L,\beta'}$, it follows that
  \be\label{eq gap ub}
  |\langle H_{\Lambda_L}\rangle_{\Lambda_L,\beta,n,u}
  -\langle H_{\Lambda_L}\rangle_{\Lambda_L,\beta',n,u}|\leq
  K_L\frac{\beta L}{8} \e{-C\beta}.
  \ee
  Here $C>0$ is uniform in $\beta,\beta',L$, while $K_L$ depends on
  $H_{\Lambda_L}$ and thus on $L$, but not on $\beta$ or $\beta'$.
  Now let $e_0(L)<e_1(L)<\dotsb<e_r(L)$ denote the distinct eigenvalues
  of $H_{\Lambda_L}$ and $m_0(L),m_1(L),\dotsc,m_r(L)$ their
  multiplicities, where $r\leq n^{2L+1}-1$.
  Taking the trace in an orthonormal basis of egenvectors, we have 
  \be
  \langle H_{\Lambda_L}\rangle_{\Lambda_L,\beta,n,u}=
  \frac{e_0(L)+\sum_{j=1}^r \frac{m_j(L)}{m_0(L)}e_j(L)\e{-\beta(e_j(L)-e_0(L))}}
  {1+\sum_{j=1}^r \frac{m_j(L)}{m_0(L)}\e{-\beta(e_j(L)-e_0(L))}}.
  \ee
  Using the same expression for $\beta'$ and letting $\beta'\to\oo$ in
  \eqref{eq gap ub} gives
  \be\begin{split}
    K_L\frac{\beta L}{8} \e{-C\beta}&\geq 
    |\langle H_{\Lambda_L}\rangle_{\Lambda_L,\beta,n,u}
    -e_0(L)|\\
    &=\e{-\beta\Delta_L}
    \frac{m_1(L)\Delta_L+\sum_{j=2}^r m_j(L) (e_j(L)-e_0(L))\e{-\beta(e_j(L)-e_0(L))}}
    {m_0(L)+\sum_{j=1}^r {m_j(L)}\e{-\beta(e_j(L)-e_0(L))}}
  \end{split}\ee
  where $\Delta_L=e_1(L)-e_0(L)$ is the gap of the finite chain.
  As $\beta\to\oo$, the fraction on the right converges
  (to  $\frac{m_1(L)}{m_0(L)}\Delta_L$), so by taking logarithms 
  and letting $\beta\to\oo$ it follows that 
  $\Delta_L\geq C>0$.
\end{proof}

	\section{Dimerization and exponential decay
		in the loop model}\label{sec 1d}

	\subsection{Clusters and repair}
	\label{ssec clusters}
	
	We now focus on the loop model. As noted above, our proof is inspired
	by that of \cite{DCPSS} for the loop $O(n)$ model on the hexagonal
	lattice. The main goal is to prove a ``repair map lemma'' (our
	Proposition \ref{prop:repair}). Heuristically, this shows that a large
	region of a primal domain without primal, short loops is exponentially
	unlikely in the size of the region (and uniformly in the size of the
	domain). The proof is based on mapping the configuration in such a
	region to one with primal, short loops everywhere in that region, this
	increasing the probability (or energy) of the configuration. This map
	is called the repair map. One has to then check how many pre-images
	the map has (the entropy), and show that the energy gain outweighs the
	entropy loss.  
	
	The proofs of Theorems \ref{thm hole perimeter} and
	\ref{qrepair:thm:convergence} are essentially strengthened versions of
	the repair map Proposition \ref{prop:repair}. Unlike in \cite{DCPSS}
	we cannot use Proposition \ref{prop:repair} as an input in their
	proofs, for technical reasons noted below; we are forced to run
	stronger versions of its proof.\\

	We
	introduce some terminology which is illustrated in Figure
	\ref{fig clusters}. 
	Recall that a primal domain $\c D_\Gamma$ is bounded by a rectangular
	circuit $\Gamma$ whose vertical segments are in dual (white) columns.

	\begin{figure}[th]
		\begin{center}
			\includegraphics[scale=1]{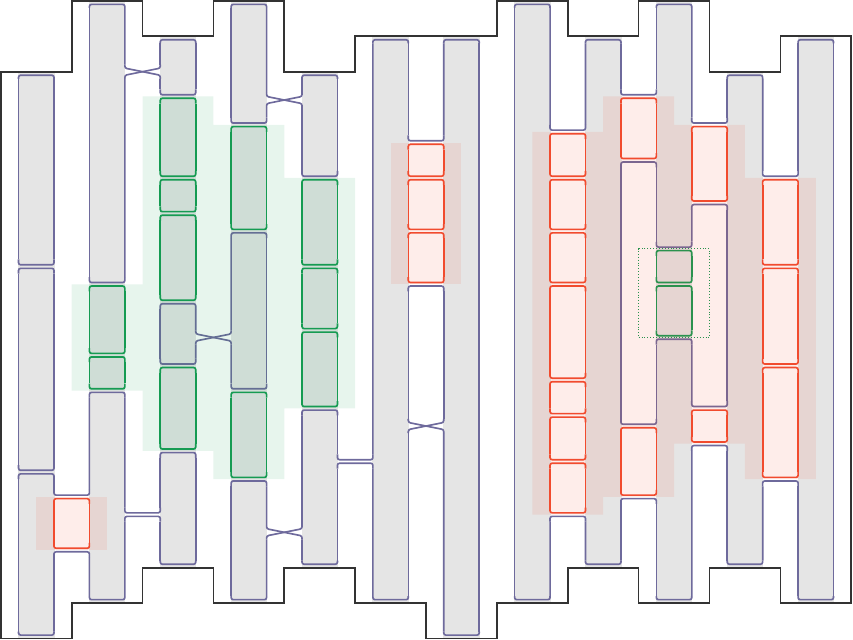}%
		\end{center}
		\caption{
			A configuration $\omega$ in a primal domain
			$\c{D}_\Gamma$.
			Long loops are drawn off-blue, small primal loops green, and small
			dual loops orange.
			Primal clusters are shaded green while dual
			clusters are shaded orange.  The small dual cluster in the lower left
			coincides with the support $\ol{\f{o}}$ of a single small dual loop $\f{o}$.
			The large dual cluster on the right contains a primal
			garden (dotted outline).
			The outside $\c O(\om)$ is the non-shaded region.
			The links strictly within this region form $\om^\out_\ast$, while
			$\om^\out$ additionally consists of double-bars at the boundary of this
			region. 
			In the upper part of the leftmost primal column there are two tall
			loops separated by a covered link.
			Note that loops do not cross cluster boundaries,
			and that all loops in $\c O(\omega)$ are long.
		}
		\label{fig clusters}
	\end{figure}

	\begin{figure}[ht]
		\begin{center}
			\includegraphics[scale=1]{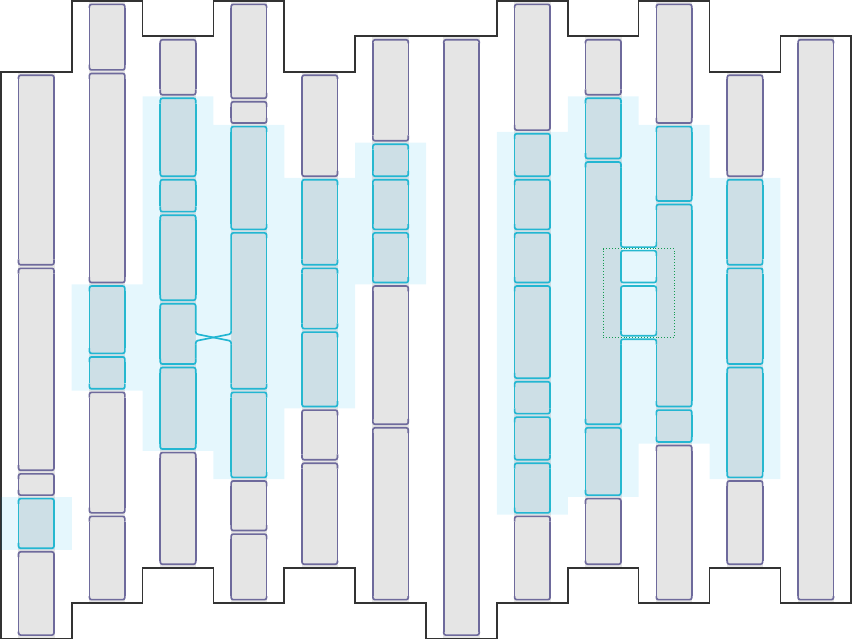}%
		\end{center}
		\caption{
			The repaired version $\bar\omega$ of the configuration
			$\omega$ in Figure \ref{fig clusters} with $\f C(\bar\omega,\bar\eta)$
			shaded turquoise.   The links on the boundary of 
			$\f C(\bar\omega,\bar\eta)$ constitute $\bar\eta$.
			Links belonging to 
			$\bar\omega^{\out}_*$ are drawn off-blue.
			Note that all loops in 
			$\c O(\bar\omega,\bar\eta)$ are trivial, primal loops.
			Tall primal loops in $\c O(\om)$ remain unchanged (e.g.\ the top of
			the leftmost column), while some additional tall loops, and thus
			covered links, may be created, as in the rightmost column and in the
			middle. 
		}
		\label{fig repair}
	\end{figure}

	Consider a configuration $\omega$ in $\c{D}_\Gamma$, where $\Gamma$ is a
	primal or dual circuit.  
	Recall that we define a trivial loop as a loop in
	the configuration $\omega$ which visits only two double bars $\dbar$ 
	spanning a single edge $e$, and that 
	a trivial loop is called primal (respectively dual) if it is in a
	primal (respectively dual) column. 
	We define a small loop as a trivial loop with 
	vertical height $<\frac{1}{\k n}$,
	where $\k\geq0$ is a constant which remains to be specified. 
	A loop which is not small is called long
	and a long trivial loop is called tall.
	
	We define a \emph{garden} as the area enclosed by
	a circuit of  small loops;
	the garden is called primal (respectively, dual) if these small
	loops are primal (respectively, dual).
	More formally,
	if $\f{o}$ is a trivial loop traversing the vertical
	intervals $\{x\}\times[r,s]$ and $\{x+1\}\times[r,s]$ (and traversing
	double-bars $\dbar$ at $(x+\tfrac12,r)$ and $(x+\tfrac12,s)$),
	let us define its \emph{support}
	$\ol{\f o}:=[x-\frac12,x+\frac32]\times[r,s]\se\bb R^2$.
	Then two trivial loops $\f o_1\neq\f o_2$ are 
	adjacent if $\ol{\f o}_1\cap\ol{\f o}_2\neq\es$.
	If $\f o_1,\dotsc,\f o_k$ are small loops forming a
	closed circuit under this adjacency-relation,
	the corresponding  garden $\f g$ is the union of
	$\ol{\f o}_1\cup\dotsb\cup\ol{\f o}_k$ with the finite component of
	$\bb R^2\sm (\ol{\f o}_1\cup\dotsb\cup\ol{\f o}_k)$.
	We define a \emph{cluster}  as a garden
	which is maximal in the sense that is not contained in any other
	garden;  clusters are similarly primal or dual.
	
	Given $\omega$, write $\f C^1=\f C^1_\k(\omega)$ for the union of its primal
	clusters, $\f C^2=\f C^2_\k(\omega)$ for the union of its dual
	clusters and $\f C(\om)=\f C^1(\om)\cup\f C^2(\om)$.
	\phantomsection\label{not clusters}%
	We define the \emph{outside}
	$\c O(\om)=\c O_\k(\omega)=
	\c{D}_\Gamma\sm\f C(\omega)$.
	\phantomsection\label{not outside}%
	Note that $\f C(\om)$ is defined as a closed subset of $\bb R^2$ and
	$\c O(\om)$  as an open subset.
	By definition, any loop in $\c O(\omega)$ is long.
	Recall that the \emph{volume} $\vol(\c O)$
	\phantomsection\label{not vol outside}%
	is defined as the total
	(vertical) length of the intervals of
	$\c O\cap(\bb Z\times \RR)$.  Thus,
	$\vol(\c O)$ coincides with the total length of all the (long) loops
	in $\c O$.
	
	We define  $\omega^{\out}_\ast=\omega\cap \c O(\omega)$,
	the restriction of $\omega$ to  $\c O$.
	\phantomsection\label{not omega-out}%
	Since we defined $\c O(\om)$ as an open set, 
	$\om^\out_\ast$ consists only of
	links `strictly' within $\c O$.
	We further define $\om^\out$ by adding to $\om^\out_\ast$
	all links on the boundary of the clusters $\f C$, 
	as well as a double-bar at each point
	where $\Gamma$ crosses a primal column if $\Gamma$ is primal
	(respectively, where $\Gamma$ crosses a dual column if $\Gamma$ is
	dual). 
	A link of $\omega^{\out}$ is called \emph{covered}
	if it is a double-bar $\dbar$, of which one half is adjacent to a tall
	loop and the other is adjacent to either another tall loop or to the
	boundary of $\c O$ (the latter case occurs if the link 
	lies on $\Gamma$ or on the boundary of $\f C$).
	A link of $\omega^{\out}$ which is not covered
	is called \emph{exposed}.
	We write $\omega^{\ex}$ for the set of exposed links
	and $\om^\ex_\ast=\om^\ex\cap\om^\out_\ast$.  
	Note that if $\k=0$ then all trivial loops are considered
	small, and $\om^\ex=\omega^{\out}$.

	Henceforth we focus on the case when $\Gamma$ is a primal 
	circuit and as before we write $\PP_{\Gamma,n,u}^1$ for  the distribution
	of $\omega$ in $\c{D}_\Gamma$ with primal boundary condition.
	Intuitively, for large $n$ we expect $\c{D}_\Gamma$ to be dominated by
	primal clusters if $\Gamma$ is primal, respectively dual clusters if
	$\Gamma$ is dual.  The alternative is that the outside $\c O$ forms a
	barrier between the boundary $\Gamma$ and the centre of the domain.  
	The main objective of this section is to prove the following:
	
	\begin{proposition}[Repair map]\label{prop:repair}
		For any $u\in[0,1)$ there 
		is a $\k_0(u)>0$ such that for
		$\k\in[0,\k_0]$ there are constants
		$C=C(u,\k)>0$ and $n_0=n_0(u,\k)<\oo$, such that
		the following holds.  For all $n>n_0$, all $v>1$,
		any primal circuit $\Gamma$, and any $x_0\in \c{D}_\Gamma$,
		\be
		\PP_{\Gamma,n,u}^1[\c O_\k\ni x_0, \vol(\c O_\k)>v] \le \e{-Cnv}.
		\ee
	\end{proposition}
	

	For the proof of Proposition \ref{prop:repair}
	we compare configurations $\omega$ with a large outside
	$\c O(\omega)$ to `repaired' configurations where most of
	$\c O$ is instead taken up by small primal loops.
	To carry out the argument,
	we split into the cases when $|\omega^{\ex}|$ is `large' or
	`small', respectively.
	Intuitively, if $|\omega^{\ex}|$ is large
	then we get a large gain in likelihood (due to an
	increased number of loops) after repair.  If $|\omega^{\ex}|$
	is small,  however, the increase in number of loops is too 
	small to be useful;
	instead, we show that $|\omega^{\ex}|$ is unlikely to be
	small, essentially because it should behave like a Poisson process of
	rate $n$.

	In what follows, for $\delta>0$ we write
	\be
	A_v^{\geq\delta nv} =
	\{\om:\c O_\k(\om)\ni x_0, 
	\vol(\c O_\k(\om))\leq v,
	|\om^\ex|\ge\delta nv\}
	\ee
	and
	\be
	A_v^{<\delta nv} =
	\{\om:\c O_\k(\om)\ni x_0, \vol(\c O_\k(\om)) \in(v-1,v],
	|\om^\ex|<\delta nv\}.
	\ee
	
	\begin{lemma}[Many exposed links]\label{lem repair dense}
		Let $u\in[0,1)$, $\delta>0$ and $\kappa\geq0$.
		For any $C>0$, there is $n_0$ depending only on $u$,
		$C$ and $\delta$ such
		that for $n>n_0$
		\be\label{eq:repairlem-0} 
		\PP^1_{\Gamma,n,u}[A_v^{\geq\delta nv}]\leq \e{-C nv}.
		\ee
	\end{lemma}

	\begin{lemma}[Few exposed links]\label{lem repair sparse}
		Let $u\in[0,1)$.
		For $\delta>0$ small enough, there is a constant 
		$C=C(u,\delta)>0$ and $n_0=n_0(u,\delta)$  such that for all
		$n>n_0$ and all $\k\in[0, \delta]$ we have
		\be
		\PP^1_{\Gamma,n,u}[A_v^{<\delta nv}]\leq \e{-C nv}.
		\ee
	\end{lemma}

	Note that the $C$ in the exponent 
	in Lemma \ref{lem repair sparse}  is
	fixed, while the one in Lemma \ref{lem repair dense} can be taken
	arbitrarily large (by taking $n$ large). 
	
	Given these two lemmas, we prove Proposition \ref{prop:repair}
	as follows.

	\begin{proof}[Proof of Proposition \ref{prop:repair}]
		Write
		$A_v=\{\om:\c O_\k(\om)\ni x_0, \vol(\c O_\k(\om))\in(v-1,v]\}$.
		Combining Lemmas \ref{lem repair dense}
		and \ref{lem repair sparse}, with $\delta>0$ chosen small enough and
		$\k\in[0,\delta]$,   summing over $v$ we get
		\be
		\PP_{\Gamma,n,u}^1[\c O_\k\ni x_0, \vol(\c O_\k)\geq v] \le 
		\sum_{w=v}^\infty \PP_{\Gamma,n,u}^1[A_w]
		\le \e{-Cnv},
		\ee
		for some constant $C=C(u)>0$.
	\end{proof}

	\subsubsection{The repair map and basic tools}
	
	Let us describe the proof strategy for
	Proposition \ref{prop:repair} (briefly described at the start of
	Section \ref{ssec clusters}).
	We will compare a configuration $\omega$ 
	belonging to the event 
	$\{\c O\ni x_0, \vol(\c O)\geq v\}$
	with a \emph{repaired} configuration $R(\omega)$.  
	If the repaired configuration
	is sufficiently more likely than the original configuration, and there
	is not too much loss of information in the repair map $R$, then it
	will follow that the event we started with was unlikely.
	The gain in likelihood will be obtained by defining $R$ so that
	$R(\omega)$ typically has significantly more loops than $\omega$,
	thereby boosting the weight factor $n^{\#\r{loops}}$.  However, it is
	essential to also control the number of possible preimages of a given
	repaired configuration $R(\omega)$ and to show that it does
	not offset the gain in likelihood.  This part of the argument
	is significantly harder in our situation than in
	the discrete setting of \cite{DCPSS},
	essentially because the continuous nature of our model 
	allows for the loss of information to be on an arbitrarily
	larger scale than the gain in likelihood.  
	We will deal with this by identifying `bad' configurations and
	bounding their probability.
	
	We define the repair map $R$ as follows,  see
	Figure \ref{fig repair} for an
	illustration.  From the configuration 
	$\om$,  form a new configuration $\bar\omega$ by
	shifting the dual clusters $\f C^2(\om)$ (and all links in them)
	left one step and,  in $\c O(\om)$,
	shifting  any links on dual  columns left one step, as well as 
	changing any crosses $\cross$ to double  bars $\dbar$. 
	(The primal clusters $\f C^1(\om)$ are kept fixed,
	and inside the clusters no links are changed). 
	We write $\eta\se\om$ for the links which lie on the
	boundaries of 
	the clusters in the original configuration $\omega$, and we let
	$\bar\eta\se\bar\omega$ denote the image of $\eta$ under the operations
	described above.  Then we define 
	$R(\omega)=(\bar\omega,\bar\eta)$.
	We will refer to $\bar\om$ as the \emph{repaired configuration}.
	
	
	For $(\bar\omega,\bar\eta) =R(\omega)$, let 
	$\bar\omega^{\out}$ and $\bar\omega_\ast^{\out}$
	denote the images of
	the respective sets $\omega^{\out}$ and $\omega_\ast^{\out}$.
	Let 
	$\f C(\bar\omega,\bar\eta)$ be the union of the
	regions bounded by $\bar\eta$ (images of the clusters), and let 
	$\c O(\bar\omega,\bar\eta)=\c{D}_\Gamma\sm\f C(\bar\omega,\bar\eta)$.  
	Note that  $\bar\omega^{\out}$, $\bar\omega^{\out}_\ast$, 
	$\f C(\bar\omega,\bar\eta)$ 
	and $\c O(\bar\omega,\bar\eta)$ can be uniquely
	reconstructed from the pair $(\bar\omega,\bar\eta)$.
	Similarly to $\om^\out$, we call a  link of $\bar\omega^{\out}$  
	\emph{covered} if it is a double-bar of which one half is adjacent to
	a tall loop and the other is adjacent to  another tall loop or to the
	boundary of $\c O(\bar\om,\bar\eta)$, and 
	a link of $\bar\omega^{\out}$ which is not covered
	is called \emph{exposed}.
	We write  $\bar\omega^{\ex}$ for the exposed links of $\bar\om$.
	
	We now make 
	the following  observations about the mapping $R$.
	\begin{enumerate}[leftmargin=*,label=\arabic*.]
		\item  If $(\bar\omega,\bar\eta) =R(\omega)$ for some $\omega$,
		then
		in order to reconstruct $\om$ from the pair $(\bar\omega,\bar\eta)$
		it suffices to know for each link in $\bar\omega^{\r{out}}$ 
		whether it was shifted or not and whether
		it was changed from a cross to double-bar or not.  Thus
		\be\label{eq preim}
		|R^{-1}(\bar\omega,\bar\eta)|\leq 4^{|\bar\omega^{\r{out}}|}.
		\ee
		\item Under the repair map all the non-trivial
		loops in $\c O(\om)$ become trivial loops
		(while the number of loops inside
		clusters and the number of tall primal loops does not change).
		Thus, the total number of loops increases;
		in fact,
		\be\label{eq loop inc}
		\ell(\bar\omega)-\ell(\omega)\geq 
		\tfrac14 |\omega^{\ex}|\geq \tfrac14 |\bar\omega^{\ex}|.
		\ee
		To see this, recall that $\om^\ex_\ast$ are the exposed links strictly
		in $\c O(\om)$, thus $\om^\ex\sm\om^\ex_\ast$
		are the exposed links on the boundary of $\c O(\om)$.
		Before repair every non-trivial loop in $\c O(\om)$ traverses
		exposed links at least 4 times each, so the number of such loops is at
		most $\frac14(2|\om^\ex_\ast|+|\om^\ex\sm\om^\ex_\ast|)$.
		After repair, all loops in $\c O(\bar\om,\bar\eta)$ are trivial and
		traverse links exactly 2 times each, thus there are exactly
		$\frac12(2|\om^\ex_\ast|+|\om^\ex\sm\om^\ex_\ast|)$ of them.
		This gives the first inequality in \eqref{eq loop inc}.
		To  see the second inequality, note that any covered  link in
		$\om^\out$ is mapped to a covered link of $\bar\om^\out$, while some
		exposed links of $\om^\out$ can be mapped to covered links of
		$\bar\om^\out$ (for example the two links on the boundary $\Gamma$ in
		the rightmost column of Figures \ref{fig clusters} and \ref{fig repair}).
		\item Let $\vol(\c O(\bar\omega,\bar\eta))$
		denote the total length of the columns in 
		$\c O(\bar\omega,\bar\eta)$,
		and $\vol(\c O^1(\bar\omega,\bar\eta))$
		the total length of  primal (grey) columns.  Then
		\be\label{eq odd area}
		\vol(\c O^1(\bar\omega,\bar\eta))\geq
		\tfrac12\vol(\c O(\bar\omega,\bar\eta))
		\geq\tfrac12 \vol(\c O(\omega)).
		\ee
		Indeed, the first inequality holds since 
		if a point in a dual (white) column belongs to 
		$\c O(\bar\omega,\bar\eta)$ then the point in the primal column to its
		left also belongs to $\c O(\bar\omega,\bar\eta)$, while 
		the second inequality holds since no area is taken away when the dual
		clusters are shifted.  
	\end{enumerate}
	
	Recall that $\PP_1$ is the law of a Poisson process of intensities
	$u,1-u$ and  write 
	$Z_\Gamma(n)=\EE_1[n^{\ell(\omega)}]$ 
	for the partition function in the primal domain $\c D_\Gamma$.
	The following is a key lemma which allows to compare the probability
	of an event $A$ with its repaired version
	$R(A)=\{R(\om):\om\in A\}$.
	
	\begin{lemma}\label{lem repair}
		Let $\hat u=\tfrac{u}{1-u}\vee1$.
		For any event $A$ depending on the links in $\c D_\Gamma$, 
		\be\label{eq repair}
		\PP^1_{\Gamma,n,u}(A)\leq \frac1{Z_\Gamma(n)}
		\int \r d \PP_1(\bar\omega)\; n^{\ell(\bar\omega)}
		\sum_{\substack{\bar\eta\se\bar\omega: \\
				(\bar\omega,\bar\eta)\in R(A) }}
		(4\hat u)^{|\bar\omega^{\out}|}
		n^{-\frac14|\bar\omega^{\ex}|}.
		\ee
	\end{lemma}

	\begin{proof}
		Writing $\PP_n$ for $\PP^1_{\Gamma,n,u}$ and using Mecke's formula,
		Lemma \ref{lem mecke}, 
		\be\begin{split}\label{eq repair mecke}
			\PP_n(A)&=\frac1{Z_\Gamma(n)}
			\EE_1[ n^{\ell(\omega)} \one_A(\omega)]\\
			&= \frac1{Z_\Gamma(n)}\sum_{r\geq0} \int\r d\mu^{\odot r}(\eta)
			\int \r d \PP_1(\omega)\; n^{\ell(\om\cup\eta)} \one_A(\om\cup\eta)
			\one\{\eta=\partial^{\r{link}}\f C(\om\cup\eta)\}.
		\end{split}\ee
		Here we wrote $\partial^{\r{link}}\f C$ for the set of links on the
		boundary of clusters.
		Now write $(\bar\om,\bar\eta)=R(\om\cup\eta)$ and note that we have 
		$\r d\mu^{\odot r}(\bar\eta)=\r d\mu^{\odot r}(\eta)$ by symmetry,
		$n^{\ell(\om\cup\eta)}\leq
		n^{\ell(\bar\om\cup\bar\eta)}n^{-\frac14|\bar\om^\ex|}$
		by \eqref{eq loop inc}
		$\one_A(\om\cup\eta)\leq \one_{R(A)}(\bar\om\cup\bar\eta,\bar\eta)$
		and
		$\one\{\eta=\partial^{\r{link}}\f C(\om\cup\eta)\}\leq
		\one\{\bar\eta=\partial^{\r{link}}\f C(\bar\om\cup\bar\eta,\bar\eta)\}$
		by definition, and
		\be
		\r d\PP_1(\om)\leq 
		4^{|\bar\om^\out|} \big(\tfrac u{1-u}\big)^{\#\cross\text{ in }\om^\out}
		\r d\PP_1(\bar\om) 
		\leq (4\hat u)^{|\bar\om^\out|} 
		\r d\PP_1(\bar\om) 
		\ee
		by \eqref{eq preim}.
		Putting all this into \eqref{eq repair mecke} and using Mecke's formula in
		reverse, we arrive at \eqref{eq repair}.
	\end{proof}
	
	In preparation for the proofs of Lemmas \ref{lem repair dense} 
	and \ref{lem repair sparse},
	we collect here some basic
	properties of our loop-model, starting with stochastic domination.
	We say that an event $A$ is
	\emph{increasing} if 
	$\om\in A$ and $\om'\supseteq\om$ imply that $\om'\in A$.  Write
	$\PP_{a,b}^{\r{Poi}}$ for the probability measure under which $\omega$ is a Poisson
	process of intensities $a$ and $b$ for $\cross$ and $\dbar$ respectively.
	\phantomsection\label{not Pn}%
	
	\begin{lemma}\label{lem:stoch dom}
		Let $\PP_n$ denote the distribution of the loop model with any
		boundary condition.  For any increasing event $A$
		we have
		$\PP_{n}(A)\leq \PP_{un,(1-u)n}^{\r{Poi}}(A)$.
	\end{lemma}
	
	\begin{proof}
		Since the number of loops changes by $\pm1$
		if a link is added or removed,  the result  follows from 
		\cite[Theorem~1.1]{georgii-kuneth}.
	\end{proof}
	
	This lemma will usually be applied to events that do not depend on the
	types $\cross$ and $\dbar$ of  the links, only on their coordinates,
	for which $\smash{\PP_{un,(1-u)n}^{\r{Poi}}}$ may be regarded as an unmarked 
	Poisson process $\PP_n^{\r{Poi}}$ of intensity $n$.
	
	We also have stochastic domination from below by a
	Poisson process of intensity $1/n$.  However, this lower bound
	will not be useful for us,
	in fact the stochastic upper bound in Lemma \ref{lem:stoch dom}
	is in some sense sharp for large $n$.  
	Intuitively, this is because we expect mostly
	small loops gathered on alternating columns,
	and then
	the number of loops $\ell(\omega)$ and the number of links $|\omega|$
	are roughly the same, meaning that the weight factor 
	$n^{\ell(\omega)}$ roughly equals 
	$n^{|\omega|}$, the latter being the weight factor for 
	an intensity $n$ Poisson process.
	We will use the following rigorous version
	of this intuition.  The proof is a simple application of Bayes'
	formula. 
	
	\begin{lemma}\label{lem cond poisson}
		Let $\Gamma$ be a primal circuit and let $T^1$
		\phantomsection\label{not Eo}%
		denote the event that $\om$ 
		consists of only double-bars located on primal columns.  
		Then 
		$\PP_{\Gamma,n,u}^1(\cdot\mid T^1)=
		\PP^{\r{Poi}}_{0,(1-u)n}(\cdot\mid T^1)$.
	\end{lemma}
	
	In words, conditional on $T^1$ the loop-configuration is
	defined by independent Poisson processes of double-bars of intensity
	$(1-u)n$ in the primal columns only.

	We will also use the following standard large-deviations estimates for
	binomial- and Poisson random variables.
	
	\begin{lemma}[Large deviation estimates] \phantom{mmm}
		\label{lem:ld}
		\begin{itemize}[leftmargin=*]
			\item  Let $X$ be Poisson distributed with mean $\rho$.  Then
			\be\label{eq ld poisson}
			\PP(X> K\rho)\leq
			\e{-\rho K\log(K/\r e)}
			\qquad\text{and}\qquad
			\PP(X< \eps\rho)\leq
			\e{-\rho[1-\eps-\eps\log(\frac1\eps)]}
			\ee 
			\item Let $Y$ have binomial distribution Bin($m,p$)
			and let $q\in(0,1)$.  Then
			\be\label{eq ld bin}
			\PP(Y> (1-q)m)\leq \exp\big(-m[
			q\log(\tfrac q{1-p})+(1-q)\log(\tfrac {1-q}{p})
			]\big).\\
			\ee
		\end{itemize}
	\end{lemma}

	\subsubsection{Proof of Lemma \ref{lem repair dense}}
	
	We turn to the upper bound on $\PP_n(A_v^{\geq\delta nv})$.
	Here, as well as in later arguments, we will use
	a discretization of the outside $\c O_\k(\om)$ into what we call a
	\emph{block-outside}.  Given $h>0$, 
	define \emph{blocks}
	\be\label{eq blocks}
	b_{i,j}:=\{2i+1,2i+2\}\times[j\tfrac hn,(j+1)\tfrac hn],
	\qquad i\in \ZZ, j\in\ZZ.
	\ee
	Thus, blocks have height $h/n$ and they span two columns: 
	one primal (grey) 
	and one dual (white), the primal column being to the left.
	The total length, or volume, of a block is therefore $2h/n$.
	We divide $\c{D}_\Gamma$ into the blocks
	$b'_{i,j}:=b_{i,j} \cap \c{D}_\Gamma$  which are non-empty.
	We refer to the intervals 
	$\{2i+1\}\times[j\tfrac hn,(j+1)\tfrac hn]$
	and
	$\{2i+2\}\times[j\tfrac hn,(j+1)\tfrac hn]$ as the left and right
	columns of $b_{i,j}$ respectively, and use the same terminology for
	$b'_{i,j}$. 
	Define the \emph{block-outside}
	$\r b_h\c O_\k(\omega)$ as the union of those blocks 
	$b'_{i,j}$ which intersect $\c O(\om)=\c O_\k(\om)$  non-trivially.   
	
	Write $N_h(\om):=|\om\cap\r b_h\c O(\omega)|$ for the number of links
	in the block-outside.  
	We claim that
	if $(\bar\omega,\bar\eta)=R(\omega)$ then
	$N_h(\om)=|\bar\om\cap\r b_h\c O(\omega)|$
	i.e.\ the repair map does not alter the number of links in the
	block-outside.  
	In fact, for each block
	$b_{i,j}'\se\r b_h\c O$ we have  
	$|b'_{i,j}\cap\bar\om|=|b'_{i,j}\cap\om|$.
	Indeed, the only links which are moved by the repair map are those in
	dual columns of the outside $\c O$, and those in dual clusters.  Links in
	dual columns of $\c O$ are shifted one step left and thus remain  in
	the same block $b'_{i,j}$.  And for a link in a dual cluster to shift
	into or out of $b_{i,j}'\se\r b_h\c O$, 
	the link would have to lie on the left or right
	boundary of the cluster;  but the vertical boundary of a cluster
	(indeed, garden) does not contain any links, by definition.
	Since the block-outside $\r b_h\c O$ contains the outside 
	$\c O$, it follows that
	\be\label{eq:block boundary size}
	|\bar\om^\out|\leq N_h(\om), \qquad
	\text{ for } (\bar\om,\bar\eta)=R(\om)
	\text{ and any }h>0.
	\ee
	
	The number of possible block-outsides $\r b_h\c O$
	can be bounded using standard
	arguments from graph theory.  Indeed,  form a graph whose
	vertices are the blocks $b'_{i,j}$ with an edge between two blocks
	if they are adjacent horizontally or vertically.  This graph has
	maximum degree 4, and for each $\omega$, the block-outside 
	$\r b_h\c O(\om)$ corresponds to a \emph{connected} subgraph.
	For a graph of maximum degree $d$, the number of connected subgraphs
	of $m$ vertices containing a given vertex is at most $(d^2)^m$
	\cite[Ch.\ 45]{bollobas}.  Thus, for any  $x_0\in \c{D}_\Gamma$, writing
	$\#\r b_h\c O(\omega)$ for the number of blocks,
	\be\label{eq graph theory}
	\#\{\r b_h\c O(\om): 
	x_0 \in \r b_h\c O(\om),
	\#\r b_h\c O(\omega)=m\}\leq 16^m.
	\ee
	
	\begin{proof}[Proof of Lemma \ref{lem repair dense}]
		We introduce the following two `bad' events:
		\be\begin{split}
			B_1(v,\eps)&=\{\om:
			x_0\in\c O,\vol(\c O)\leq v,\#\r b_\eps\c O>\tfrac1\eps nv\}\\
			B_2(v,\eps,L)&=\{\om:
			x_0\in\c O,\vol(\c O)\leq v, N_\eps(\om)>L nv\}.
		\end{split}\ee
		Here we think of $\eps>0$ as small and $L>0$ as large,
		thus $B_1$ is the event that the block-outside has very many blocks,
		while $B_2$ is the event that it contains very many links.  
		We have that
		\be
		\PP_n(A_v^{\geq \delta nv})\leq \PP_n(B_1)+
		\PP_n(B_2\sm B_1)+
		\PP_n(A_v^{\geq \delta nv} \sm (B_1\cup B_2))
		\ee
		and we proceed by bounding each of these three terms.  The first two
		will be bounded using stochastic domination, while the last will be
		bounded using Lemma \ref{lem repair}.
		
		First consider $B_1$.  We claim that
		for any $\eps\in(0,2^{-10})$,
		\be\label{eq:Bv and B1}
		\PP_n(B_1(v,\eps))
		\leq \frac{\exp(-nv\cdot\tfrac1 {2\eps}\log(\tfrac1{2^{10}\eps}))}
		{1-32\sqrt{\eps}}.
		\ee
		To see this, first note that on $B_1(v,\eps)$, at
		least half of the blocks constituting $\r b_\eps\c O$ contain one or more link
		each.  Indeed, if not
		then at least half the blocks constituting $\r b_\eps\c O$ contain no link.
		Any such block is fully contained in $\c O$, which has volume at most
		$v$.  So,  writing $m$ for the number of blocks in $\r b_\eps\c O$,
		\be
		v\geq \tfrac12 m\cdot\tfrac{2\eps}{n}>
		\tfrac12 \big(\tfrac1\eps nv\big)\cdot\tfrac{2\eps}{n}=v,
		\ee
		a contradiction.  It follows that on $B_1(v,\eps)$,
		there is some connected component of blocks, containing $x_0$ and
		consisting  of at least  $m_1={\frac1\eps nv}$  
		blocks, such that at least half of its
		blocks contain one or more links each.   The number of choices
		of such a component with $m$ blocks is at most $16^m$, 
		by \eqref{eq graph theory},
		and for a given
		such component, the number of blocks containing a link is
		stochastically dominated by a Bin($m$,$1-\r e^{-\frac{2\eps}{n}\cdot n}$)
		random variable, by Lemma \ref{lem:stoch dom}.  Noting that 
		$1-\r e^{-\frac{2\eps}{n}\cdot n}\leq 2\eps$, and using large
		deviations estimates \eqref{eq ld bin} with $q=\frac12$ and $p=2\eps$,
		it follows that 
		\be
		\PP_n(B_1(v,\eps))
		\leq 
		\sum_{m\geq m_1} 16^m \exp\big(-\tfrac m2
		\log(\tfrac{1}{8\eps})\big)
		\leq \frac{\exp(-nv\cdot\tfrac1 {2\eps}\log(\tfrac1{2^{10}\eps}))}
		{1-32\sqrt{\eps}}.
		\ee
		
		Next consider $B_2\sm B_1$.  
		On this event there is some connected component of at most
		$m_1=\frac1\eps nv$
		blocks which contains $>Lnv$ links (and contains $x_0$).
		The number of choices of such a connected component of $m$ blocks is
		at most $16^m$, by \eqref{eq graph theory},
		and for a fixed such component, the event that it
		contains    $>Lnv$ links is increasing.  Hence, by Lemma
		\ref{lem:stoch dom}
		\be \label{eq B2 not B1 1}
		\PP_n( B_2\sm B_1)\leq
		\sum_{m=1}^{m_1} 16^m \PP(X>Lnv),
		\ee
		where $X$ is a Poisson distributed random variable with mean 
		$m_1\cdot 2\eps/n\cdot n=2nv$.  Using 
		large deviations  \eqref{eq ld poisson}, it follows
		that 
		\be\begin{split}\label{eq B2 not B1}
			\PP_n( B_2(v,\eps,L)\sm B_1(v,\eps))&\leq
			m_1 16^{m_1} \exp(-L\log(\tfrac L{\r e})nv)
			=\tfrac{nv}{\eps}  \exp(-nv[-\tfrac{\log 16}{\eps}+L\log(\tfrac L{\r e})]).
		\end{split}\ee
		
		Finally consider $A_v^{\geq \delta nv} \sm (B_1\cup B_2)$.
		We start by summing over the possibilities $\lambda$
		for the block-outside $\r b_\eps\c O$:
		\be
		\PP_n(A_v^{\delta nv} \sm (B_1\cup B_2))\leq
		\sum_{\lambda:\#\lambda\leq m_1}
		\PP_n(A_{v,\lambda}^{\geq\delta nv} \sm B_2)
		\ee
		where 
		$A_{v,\lambda}^{\geq\delta nv} =A_{v}^{\geq\delta nv}\cap\{ \r b_\eps\c O=\lambda\}$.
		Using Lemma  \ref{lem repair} we get
		\be\begin{split}\label{eq dense 2}
			\PP_n(A_{v,\lambda}^{\geq\delta nv} \sm B_2)&
			\leq \frac1{Z_\Gamma(n)}
			\int \r d \PP_1(\bar\omega)\; n^{\ell(\bar\omega)}
			\sum_{\substack{\bar\eta\se\bar\omega: \\
					(\bar\omega,\bar\eta)\in R(A_{v,\lambda}^{\geq\delta nv} \sm B_2) }}
			(4\hat u)^{|\bar\omega^{\out}|} n^{-\frac14|\bar\om^\ex|}\\
			&\leq
			(4\hat u)^{L nv}n^{-\frac14\delta nv}
			\frac1{Z_\Gamma(n)}
			\int \r d \PP_1(\bar\omega)\; n^{\ell(\bar\omega)}
			\#\{\bar\eta\se\bar\omega: 
			(\bar\omega,\bar\eta)\in R(A_{v,\lambda}^{\geq\delta nv} \sm B_2) \}\\
			&\leq (4\hat u)^{L nv}n^{-\frac14\delta nv}
			2^{Lnv}=
			(8\hat u)^{Lnv}\e{-(\frac\delta4\log n)nv}.
		\end{split}\ee
		We used \eqref{eq:block boundary size}
		to bound $|\bar\om^\out|\leq Lnv$ on $B_2^c$, and  
		in the last step we used that the number of choices of $\bar\eta$ is
		at most the number of subsets of $\bar\om$,
		which on $A_{v,\lambda}^{\geq\delta nv} \sm B_2$ is at most $2^{Lnv}$.
		Bounding the number of possibilities for $\lambda$ using 
		\eqref{eq graph theory}, we get
		\be
		\PP_n(A_v^{\delta nv} \sm (B_1\cup B_2))\leq
		(\tfrac1\eps nv) 16^{\frac1\eps nv} (8\hat u)^{Lnv}\e{-(\frac\delta4\log n)nv}
		\leq \e{-\frac\delta4(\log n-C)nv},
		\ee
		where $C$ depends on $u$, $\eps$ and $L$.
		Combining this with the other terms \eqref{eq:Bv and B1}
		and \eqref{eq B2 not B1}, we may first take $\eps>0$ sufficiently
		small, then $L$ sufficiently large, and finally $n$ sufficiently
		large, to obtain the claim of Lemma \ref{lem repair dense}.
	\end{proof}

	\subsubsection{Proof of Lemma \ref{lem repair sparse}}
	
	We now turn to the case of few outer links, i.e.\ the upper bound on
	$\PP_n(A_v^{<\delta nv})$.  The
	rough idea  is that, after repair, the configuration in
	$\c O$  behaves like a Poisson process of intensity $n$ which is
	unlikely to contain $<\delta nv$ links by large deviations estimates.
	(The probability that such a Poisson process contains \emph{no} links
	is $\e{-nv}$, which thus gives an upper bound on the rate of decay.)
	The entropy is controlled using that $\bar\eta$ contains 
	few links.
	
	\begin{proof}[Proof of Lemma \ref{lem repair sparse}]
		Similarly to the proof of Lemma
		\ref{lem repair dense} we use the block-outside,
		this time with (large)  $h=1/\delta$.
		Recall that a trivial loop is called tall if it has 
		height $>\frac1{\k n}$, and that we assume $\k\leq\delta$.
		
		First note that, on the event $A_v^{<\delta nv}$, there are at most
		$\k nv$ covered links in $\c O$.  Indeed,
		each covered link is adjacent to a tall loop,  and there are at most two
		covered links adjacent to the same tall loop.
		A tall loop contributes volume $>2\frac1{\k n}$ to $\c O$,
		which means that each
		covered link contributes at least $\frac1{\k n}$ to $\vol(\c O)$.
		Since $\vol(\c O)\leq v$, there are at most
		$\k nv$ covered links.
		It follows that on  $A_v^{<\delta nv}$,
		\be\label{eq outer sparse}
		|\bar\om^\out|=|\om^\out|
		\leq (\delta+\k)nv \leq 2\delta nv,
		\ee
		since by assumption the
		number of exposed links is $<\delta nv$.

		Next note that,
		on  $A_v^{<\delta nv}$, the number of blocks
		in $\r b_{1/\delta}\c O$ satisfies
		\be\label{eq blocks sparse}
		\#\r b_{1/\delta}\c O\leq (\tfrac32\delta+\kappa) nv\leq\tfrac52\delta nv=:m_1.
		\ee  
		Indeed, each block contains either no link, at least one exposed link,
		or at least one covered link.
		Each empty block contributes all of its 
		$2/\delta n$ volume to 
		$\c O$, and $\vol(\c O)\leq v$, so there can be at most $\delta vn/2$ of
		them.   There are at most $\delta nv$
		exposed links in $\c O$, so there are at most $\delta nv$ blocks with at least
		one exposed link.  And there are at most $\k nv$ covered links
		in $\c O$, so
		there are at most $\k nv$ blocks containing a covered link.
		Summing these up we get \eqref{eq blocks sparse}.
		The total volume of $\r b_{1/\delta}\c O$ is at most 
		$m_1\frac{2}{\delta n}=5v$.

		Let $C_0$ be a large constant and let
		$B^{\geq C_0nv}$ denote the event that $\r b_{1/\delta}\c O$ 
		contains at least
		$C_0nv$ links.  (Since $\r b\c O_{1/\delta}$
		can be strictly larger than $\c O$
		itself, it is possible that $A_v^{<\delta nv}$ and $B^{\geq C_0nv}$
		both occur.)  
		Write 
		\be
		A_{v,\lam}^{<\delta nv} = A_v^{<\delta nv} \cap 
		\{\om  : \r b_{1/\delta}\c O(\om)=\lam\}.
		\ee 
		We have
		\be\label{eq A <delta}
		\PP_n(A_v^{<\delta nv})\leq
		\sum_{\lambda:\#\lambda\leq m_1}\big[
		\PP_n(A_{v,\lambda}^{<\delta nv}\sm B^{\geq C_0nv})
		+\PP_n(|\omega\cap\lambda|\geq C_0nv)
		\big]
		\ee
		where $\#\lambda$ denotes the number of blocks.
		By Lemma \ref{lem:stoch dom}, $|\omega\cap\lambda|$
		is stochastically dominated by a Poisson random variable with mean
		$5nv$, so similarly to \eqref{eq B2 not B1 1} 
		and \eqref{eq B2 not B1} we get
		\be\begin{split}
			\sum_{\lambda:\#\lambda\leq m_1}
			\PP_n(|\omega\cap\lambda|\geq C_0nv)
			\leq
			\tfrac52\delta nv
			\exp\big(-[C_0\log(\tfrac{C_0}{5\r e})
			-\tfrac52\delta\log16]nv\big).
		\end{split}\ee
		Choosing $C_0$ large enough (depending on
		$\delta$) we get that for some $C_5>0$,
		\be\label{eq B lambda bound}
		\sum_{\lambda:\#\lambda\leq m_1}
		\PP_n(|\omega\cap\lambda|\geq C_0nv)\leq
		\e{-C_5nv}.
		\ee
		For the other terms in \eqref{eq A <delta} we use Lemma 
		\ref{lem repair}.  Write $A^{<\delta nv}_{v,\lam,r}$
		for $A^{<\delta nv}_{v,\lam}$ with the
		extra condition that $|\eta|=r$.  By Lemma  \ref{lem repair},
		where we bound the factor $(4\hat u)^{|\bar\omega^{\out}|}$ 
		above by $(4\hat u)^{2\delta nv}$ (using \eqref{eq outer sparse})
		and $n^{-\frac14|\bar\om^\ex|}$ by 1,  
		\be \label{eq before mecke}
		\PP_n(A_{v,\lambda}^{<\delta nv}\sm B^{\geq C_0nv})
		\leq \frac{(4\hat u)^{2\delta nv}}{Z_\Gamma(n)}
		\sum_{r=0}^{\floor{\delta nv}}
		\int \r d \PP_1(\bar\omega)\; n^{\ell(\bar\omega)}
		\sum_{\substack{\bar\eta\se\bar\omega\\ |\bar\eta|=r}} \one\{
		(\bar\omega,\bar\eta)\in R(A_{v,\lambda,r}^{<\delta nv}\sm B^{\geq C_0nv}) \}.
		\ee
		For the integral over $\bar\om$ we use Mecke's formula,
		Lemma \ref{lem mecke}, which allows us to treat $\bar\eta$ as fixed: 
		\be\begin{split}\label{eq after mecke}
			&\int \r d \PP_1(\bar\omega)\; n^{\ell(\bar\omega)}
			\sum_{\substack{\bar\eta\se\bar\omega\\ |\bar\eta|=r}} \one\{
			(\bar\omega,\bar\eta)\in 
			R(A_{v,\lambda,r}^{<\delta nv}\sm B^{\geq C_0nv}) \}\\
			&=\int \r d\mu^{\odot r}(\bar\eta)
			\int \r d \PP_1(\bar\omega)\; 
			n^{\ell(\bar\omega\cup\bar\eta)}
			\one\{ (\bar\omega\cup\bar\eta,\bar\eta)\in 
			R(A_{v,\lambda,r}^{<\delta nv}\sm B^{\geq C_0nv}) \}.
		\end{split}\ee
		
		Now we use that, given $\bar\eta$, the remaining configuration
		$\bar\omega\sm\bar\eta$ splits as 
		$\bar\omega\sm\bar\eta=\bar\omega^\out_*\cup\bar\omega_*^\ins$, where
		$\bar\omega^\out_*$ is the configuration strictly outside the domains
		enclosed by $\bar\eta$, and $\bar\omega_*^\ins$ is the configuration
		strictly inside.  
		More precisely, recall that $\f C=\f C(\bar\omega,\bar\eta)$
		denotes the images of the clusters
		under the repair-map, which is a union of
		sub-domains of $\c{D}_\Gamma$
		whose boundaries are defined by $\bar\eta$.
		Let $\Omega_\Gamma^{\bar\eta,\out}$ be the set of
		configurations in 
		$\c O(\bar\om,\bar\eta)=\c{D}_\Gamma\sm \f C(\bar\om,\bar\eta)$ 
		compatible with $\bar\eta$ and
		let $\Omega_\Gamma^{\bar\eta,\ins}$ be
		the set of configurations in the interior of  $\f C(\bar\om,\bar\eta)$ 
		compatible with $\bar\eta$.
		Then  $\bar\omega^\out_*\in \Omega_\Gamma^{\bar\eta,\out}$ 
		and $\bar\omega^\ins_*\in \Omega_\Gamma^{\bar\eta,\ins}$.
		Also note that these two configurations
		contain strictly disjoint sets of loops since loops cannot pass
		between $\f C$ and $\c O$.
		
		The indicator constraining 
		$(\bar\omega\cup\bar\eta,\bar\eta)$ can be factorized:
		\be\begin{split}
			\one\{ (\bar\omega\cup\bar\eta,\bar\eta)&\in 
			R(A_{v,\lambda,r}^{<\delta nv}\sm B^{\geq C_0nv}) \}\\
			&\quad\leq
			\one\{\bar\eta\in\c R_{\lambda,\delta,r}\}
			\one\{\bar\omega_*^\out\in \c R^\out_{\bar\eta,2\delta}\}
			\one\{\bar\omega_*^\ins\in \c R^\ins_{\bar\eta,\lambda,C_0}\}
		\end{split}\ee
		where we use the following events:
		\be\begin{split}
			\c R_{\lambda,\delta,r}&=\{\bar\eta\in\Omega_\Gamma:
			|\bar\eta|=r, \exists\, \bar\omega\text{ s.t. }
			(\bar\omega\cup\bar\eta,\bar\eta)\in 
			R(A_{v,\lambda,r}^{<\delta nv}) \}\\
			\c R^\out_{\bar\eta,C}&=
			\{\bar\omega^\out_*\in\Omega_\Gamma^{\bar\eta,\out}
			\cap T^1:
			|\bar\omega^\out_*|<
			C nv \} \\ 
			\c R^\ins_{\bar\eta,\lambda,C}&=
			\{\bar\omega^\ins_*\in\Omega_\Gamma^{\bar\eta,\ins}:
			|\bar\omega^\ins_*\cap\lambda|< C nv \}.
		\end{split}\ee
		Recall 
		here that $T^1$ is the set of configurations
		that consist only of double-bars located in primal columns.
		
		Let $\PP^{\out,\bar\eta}_1$ and $\PP^{\ins,\bar\eta}_1$ denote the
		restrictions of $\PP_1$ to $\c O$ and $\f C$ respectively.
		The right-hand-side in \eqref{eq after mecke} is bounded above by:
		\be\begin{split}\label{eq mecke out and in}
			\int \r d\mu^{\odot r}(\bar\eta)&
			\one\{\bar\eta\in\c R_{\lambda,\delta,r}\}
			\int \r d \PP^{\out,\bar\eta}_1(\bar\omega_*^\out)\; 
			n^{\ell(\bar\omega_*^\out)}
			\one\{\bar\omega_*^\out\in \c R^\out_{\bar\eta,2\delta}\}
			\int \r d \PP^{\ins,\bar\eta}_1(\bar\omega_*^\ins)\; 
			n^{\ell(\bar\omega_*^\ins)}
			\one\{\bar\omega_*^\ins\in \c R^\ins_{\bar\eta,\lambda,C_0}\}.
		\end{split}\ee
		We focus on the middle integral in \eqref{eq mecke out and in}, over
		$\bar\omega_*^\out$. This is the part which, since we have repaired the configuration, is essentially a Poisson process of intensity $(1-u)n$ on the primal columns of $\c O(\bar\om,\bar\eta)$. In particular we're working on the event that there's at most $2\delta nv$ links in $\c O(\bar\om,\bar\eta)$ by \eqref{eq outer sparse}, which we'll be able to bound by the large deviations estimates. Let's do this rigorously. Using Lemma \ref{lem cond poisson} 
		we have
		\be
		\frac{
			\int \r d \PP^{\out,\bar\eta}_1(\bar\omega_*^\out)\; 
			n^{\ell(\bar\omega_*^\out)}
			\one\{\bar\omega_*^\out\in \c R^\out_{\bar\eta,2\delta}\}
		}
		{\int \r d \PP^{\out,\bar\eta}_1(\bar\omega_*^\out)\; 
			n^{\ell(\bar\omega_*^\out)}
			\one\{\bar\omega_*^\out\in T^1\}
		}
		= \PP_{0,(1-u)n}^{\r{Poi}}(|\bar\omega^\out_*|<2\delta nv\mid T^1),
		\ee
		where on the right-hand-side $\bar\omega^\out_*$
		is a Poisson process of double-bars of intensity $(1-u)n$
		in the primal
		columns in $\c O(\bar\omega,\bar\eta)$.
		By \eqref{eq odd area}, the primal
		columns in $\c O(\bar\omega,\bar\eta)$ have total length
		at least $\tfrac12(v-1)$, meaning that $|\bar\omega^\out_*|$
		stochastically dominates a Poisson random
		variable with mean $\frac12(1-u)n(v-1)$.
		Writing
		\be
		\PP_{0,(1-u)n}^{\r{Poi}}(|\bar\omega^\out_*|<2\delta nv\mid T^1)=
		\PP_{0,(1-u)n}^{\r{Poi}}(|\bar\omega^\out_*|<C_0 nv\mid T^1)
		\frac{\PP_{0,(1-u)n}^{\r{Poi}}(|\bar\omega^\out_*|<2\delta nv\mid T^1)}
		{1-\PP_{0,(1-u)n}^{\r{Poi}}(|\bar\omega^\out_*|\geq C_0 nv\mid T^1)},
		\ee
		it follows, using the large deviation estimates \eqref{eq ld poisson}
		with 
		$\eps=\frac{4\delta}{1-u}\frac{v}{v-1}\leq \frac{5\delta}{1-u}$ 
		and $K=\frac{2 C_0}{1-u}\frac{v}{v-1}\geq \frac{2 C_0}{1-u}$,
		that 
		\be\begin{split}
			&\int \r d \PP^{\out,\bar\eta}_1(\bar\omega_*^\out)\; 
			n^{\ell(\bar\omega_*^\out)}
			\one\{\bar\omega_*^\out\in \c R^\out_{\bar\eta,2\delta}\}\\
			&\leq
			\frac{
				\e{-\frac12(1-u)n(v-1)
					[1-\frac{5\delta}{1-u}-
					\frac{5\delta}{1-u}\log \frac {1-u}{5\delta}]}
			}
			{
				1-
				\e{-\frac12(1-u)n(v-1)
					\frac{2C_0}{1-u}\log
					\frac{2C_0}{\r e(1-u)}}
			}
			\int \r d \PP^{\out,\bar\eta}_1(\bar\omega_*^\out)\; 
			n^{\ell(\bar\omega_*^\out)}
			\one\{\bar\omega_*^\out\in \c R^\out_{\bar\eta,C_0}\}.
		\end{split}\ee
		Putting this into \eqref{eq mecke out and in} and reversing
		\eqref{eq after mecke} we get from
		\eqref{eq before mecke}
		\be\begin{split}
			&\PP_n(A_{v,\lambda}^{<\delta nv}\sm B^{\geq C_0nv})\\
			&\leq
			\frac{
				\e{-\frac12(1-u)n(v-1)
					[1-\frac{5\delta}{1-u}-
					\frac{5\delta}{1-u}\log \frac {1-u}{5\delta}]}
			}
			{
				1-
				\e{-\frac12(1-u)n(v-1)
					\frac{2C_0}{1-u}\log
					\frac{2C_0}{\r e(1-u)}}
			}
			\frac{(4\hat u)^{2\delta nv}}{Z_\Gamma(n)}
			\int \r d \PP_1(\bar\omega)\; n^{\ell(\bar\omega)}
			\one\{|\bar\omega\cap\lambda|\leq 3C_0nv\}
			\#\{\bar\eta\se\bar\omega:|\bar\eta|< \delta nv\}.
		\end{split}\ee
		Here 
		we used  that
		$|\bar\eta|\leq 2\delta nv$, $|\bar\omega_*^\out|\leq C_0 nv$ 
		and $|\bar\omega_*^\ins\cap\lambda|\leq C_0 nv$ so that 
		$|\bar\omega\cap\lambda|\leq (2C_0+2\delta)nv\leq 3C_0nv$.
		Thus the last factor satisfies
		\be\label{q-repair:eq:entropybound1}
		\begin{split}
			\#\{\bar\eta\se\bar\omega:|\bar\eta|< 2\delta nv\}
			\leq
			{3C_0nv\choose 2\delta nv}
			\leq 
			\Big(\frac{3C_0\r e}{2\delta}\Big)^{2\delta nv}.
		\end{split}
		\ee
		Thus, 
		\be
		\PP_n(A_{v,\lambda}^{<\delta nv}\sm B^{\geq C_0nv})
		\leq
		\frac{
			\e{-\frac12(1-u)n(v-1)
				[1-\frac{5\delta}{1-u}-
				\frac{5\delta}{1-u}\log \frac {1-u}{5\delta}]}
		}
		{
			1-
			\e{-\frac12(1-u)n(v-1)
				\frac{2C_0}{1-u}\log
				\frac{2C_0}{\r e(1-u)}}
		}
		\Big(\frac{6C_0\r e\hat u}{\delta}\Big)^{2\delta nv}.
		\ee
		Note that $\big(\tfrac{6C_0\r e\hat u}{\delta}\big)^{2\delta}\to1$ as
		$\delta\to0$. 
		Since the number of terms in the sum over $\lambda$ in 
		\eqref{eq A <delta} is at most $m_116^{m_1}$,
		where we recall that  $m_1=\tfrac52\delta nv$, and using 
		\eqref{eq B lambda bound}, we get
		\be \label{q-repair:eq:repairlem-3}
		\PP_n(A^{<\delta nv}_{v})\leq
		\e{-C_5 nv}+m_116^{m_1}\e{-C_6 nv}
		\leq \e{-C_7 nv},
		\ee
		for some constants $C_6,C_7>0$, positive for
		$\delta$  small enough. This concludes the proof of Lemma \ref{lem repair sparse}, and therefore together with Lemma \ref{lem repair dense} completes the proof of Proposition \ref{prop:repair}.
	\end{proof}

	\subsection{Proof of Theorem \ref{thm hole perimeter}}
	
	Recall that we need to prove exponential decay of
	\be
	\PP_{\Gamma,n,u}^1[\per(\c C(x_0))>v]
	\ee 
	in $v$, where $\c C(x_0)=\c C_\k(x_0,\omega)$ is the 
	connected component of 
	$\c D_\Gamma\sm\c P_\k(\omega)$ which contains $x_0$
	and
	$\c P_\k(\omega)$ is the connected component of small primal loops
	adjacent to the boundary $\Gamma$ of $\c D_\Gamma$.
	See Figure \ref{fig component} for an illustration.
	
	Note that the inner boundary of
	$\c D_\gamma=\c C(\om)$ consists of only long loops.  Thus,
	one might think that the Theorem should follow by 
	applying Proposition \ref{prop:repair} for each possible realization of
	$\c C$.  However, the combinatorial factor arising from
	the possibilities for $\c C$
	cannot be offset by the exponential decay in Proposition \ref{prop:repair},
	because in the latter result the constant in the exponent is fixed
	and cannot be taken large enough.
	The restriction on the exponent comes in Lemma \ref{lem repair sparse},
	i.e.\ the case where $\c O$ is very sparsely populated with links,  in
	which case the entropy grows at a much smaller  rate than a-priori.
	The same logic applies to $\c C$,
	and our proof of
	Theorem \ref{thm hole perimeter} 
	will consist of pointing out the necessary
	modifications to the proof of Proposition  \ref{prop:repair}.

	\begin{proof}[Proof of Theorem \ref{thm hole perimeter}]
		In this proof we write 
		$\c O=\c O(\omega)$ for the outside of clusters in the (now random)
		domain $\c D_\gamma={\c C}$,
		and $\omega^\ex$ for the exposed links in $\c O$.
		We claim that we have the following two
		inequalities:   first, for any $\delta>0$ and any $C_1>0$,
		provided $n$ is large enough: 
		\be\label{eq cluster 1}
		\PP_n(\vol(\c O)\leq w,|\omega^\ex|\geq \delta nw)\leq \e{-C_1 nw}
		\qquad \text{ for all } w>1,
		\ee
		and second, for some $C_2>0$, provided $\k$ is small enough and
		$n$ is large enough: 
		\be\label{eq cluster 2}
		\PP_n(\vol(\c O)>w)\leq \e{-C_2 nw},
		\qquad \text{ for all } w>1.
		\ee
		These are analogous to Lemma \ref{lem repair dense}
		and Proposition \ref{prop:repair}, respectively.
		Compared to the proofs of those results,
		the only modification required
		is equation \eqref{eq graph theory}, which
		gives the bound $16^m$ for the number of possible block-outsides 
		$\r b\c O$ with $m$ blocks and containing a given point $x_0$.  
		In the present  setting, the block-outside does not contain $x_0$ but
		rather surrounds it.  By counting according to which is the rightmost
		block  along the `$x$-axis' that $\r b\c O$ contains, we get the bound 
		\be\label{eq graph theory 2}
		\#\{\r b_h\c O(\om): 
		\#\r b_h\c O(\omega)=m, \text{ surrounding } x_0\}
		\leq m16^m\leq 17^m,
		\ee
		where the last inequality holds for $m$ large enough.
		The analogs of Lemmas \ref{lem repair dense} and 
		\ref{lem repair sparse} are then proved exactly as
		before, using the bound \eqref{eq graph theory 2} for the number of
		block-outsides.  None of the arguments were sensitive to the exact
		constant in the exponential growth of block-outsides, 
		so \eqref{eq cluster 1} and
		\eqref{eq cluster 2} follow.
		
		Next, we make a change of variables in \eqref{eq cluster 1}
		and \eqref{eq cluster 2}.  To be definite, fix $C_1=1$ and
		$\delta=\frac12$, and fix $n$ large enough that
		both  \eqref{eq cluster 1}
		and \eqref{eq cluster 2} hold.   Writing $w=\frac vn$  it
		follows that 
		\be
		\PP_n(\vol(\c O)\leq \tfrac vn,|\omega^\ex|\geq \tfrac12 v)\leq 
		\e{-v},
		\qquad
		\PP_n(\vol(\c O)>\tfrac vn)\leq \e{-C_2 v},
		\qquad\text{for }v>n.
		\ee
		By adjusting the constants in the exponents, it follows that 
		\be \label{eq cluster 3}
		\PP_n(\vol(\c O)\leq \tfrac vn,|\omega^\ex|\geq \tfrac12 v)
		\leq \e{-C_3 v},\quad
		\PP_n(\vol(\c O)>\tfrac vn)\leq \e{-C_3 v},
		\qquad\text{for }v>1,
		\ee
		where now $C_3>0$ may depend on $n$.
		
		To complete the proof of the theorem, note that 
		$\per(\c C)$ equals the sum of the
		vertical and horizontal displacements of $\gamma$.
		The horizontal displacement in turn equals
		twice the number of crossings $\gamma$ makes 
		of primal columns (each such crossing has length 2).
		Moreover, $\c O$ is a connected set
		which follows  $\gamma$, so the vertical displacement of $\gamma$ is
		a lower bound on $\vol(\c O)$, while the number of primal crossings
		is a lower bound on $|\om^\out|$.
		It follows that if $\per(\c C)>v$ then either
		$\vol(\c O)>\frac vn$, or 
		$\vol(\c O)\leq \frac vn$ and
		the number of primal crossings (and hence $|\om^\out|$) is 
		$\geq \tfrac12(v-\frac vn)$.  The probability of the former event is
		bounded in \eqref{eq cluster 3}, while for the latter event we need to
		take into account that 
		some of the primal crossings may correspond to covered links.
		However,  each covered link contributes at least $\frac1{\k n}$
		to $\vol(\c O)$, so if $\vol(\c O)\leq \frac vn$
		then there can be at most $\k v$ covered links along
		$\gamma$.  Choosing $\k$ small enough, by \eqref{eq cluster 3},
		\be \label{eq perim bound}
		\PP_n(\per(\c C)>v)\leq
		\PP_n(\vol(\c O)>\tfrac vn)+
		\PP_n(\vol(\c O)\leq \tfrac vn,|\omega^\ex|\geq 
		\tfrac12 (v-\tfrac vn)-\k v)
		\leq 2 \e{-C_4 v},
		\ee
		which completes the proof of Theorem \ref{thm hole perimeter}. \qedhere
	\end{proof}
	
	We also get the following related bound on the size of union
	of components of $\c D\sm\c P_\k(\om)$ intersecting a given domain:
	
	\begin{corollary}\label{component of a domain}
		Let $\c A$ be a bounded domain and let $\c C_\k(\c A)$ be the connected
		component of $\c A\cup (\c D\sm\c P_\k(\om))$ which contains $\c A$.  
		For any $u\in[0,1)$ there is a constant 
		$\k_0=\k_0(u)>0$ such that for all 
		$\k\in[0,\k_0]$ there is 
		$n_0=n_0(u,\eps)<\oo$
		such that the following holds.  For any $n>n_0$, 
		there is a constant $C=C(u,n,\k)>0$ such that for
		all $v>1$,
		\be\begin{split}
			&\PP_{\Gamma,n,u}^1[\per(\c C_\k(\c A))>\per(\c A)+v] \le  
			\vol(\c A)\e{-Cv},\\
			&\PP^{\r{per}}_{\Lambda_L,\beta,n,u}[\per(\c C_\k(\c A))>\per(\c A)+v] 
			\le \vol(\c A) \e{-Cv},
		\end{split}\ee
		uniformly for all primal circuits $\Gamma$
		and all $L\in2\bb Z+1$ and $\beta>0$, such that $\c A$ is contained in
		the corresponding domain
		$\c D=\c D_\Gamma$ or $\c D=\Lambda_L\times[0,\beta]$.
	\end{corollary}
	
	\begin{proof}
		This follows from a very similar argument to 
		Theorem \ref{thm hole perimeter}.
		Indeed, $\c C_\k(\c A)$ is a union of $\c A$ with sets $\c C_\k(x_0)$
		(for various $x_0$) which intersect $\c A$.  In this case
		we define $\c O(\om)$ as the union of outsides of these 
		components $\c C(x_0)$.  
		When summing over
		possible block-outsides $\r b_h\c O$, we first sum over all blocks
		intersecting $\c A$, of which there are at most
		$\vol(\c A)/\frac {2h}n$, and then the possibilities for the outside
		of a component $\c C(x_0)$ containing that block.  Using 
		\eqref{eq graph theory}, this leads to a bound
		\be
		\frac{n\cdot \vol(\c A)}{2h}16^m
		\ee
		for the number of possible $\r b_h\c O$ with $m$ blocks, meaning that 
		\eqref{eq cluster 1} and \eqref{eq cluster 2} are replaced by
		\be
		\PP_n(\vol(\c O)\leq w,|\omega^\ex|\geq \delta nw)\leq 
		C_1 n \vol(\c A) \e{-C_2 nw}
		\qquad 
		\PP_n(\vol(\c O)>w)\leq 
		C_1 n \vol(\c A)\e{-C_2 nw},
		\ee
		or after adjusting the constant in the exponents,
		\be
		\PP_n(\vol(\c O)\leq w,|\omega^\ex|\geq \delta nw)\leq 
		\vol(\c A) \e{-C_3 nw}
		\qquad 
		\PP_n(\vol(\c O)>w)\leq 
		\vol(\c A)\e{-C_3 nw}.
		\ee
		If $\per(\c C(\c A))>\per(\c A)+v$ then the union of components
		$\c C (x_0)$ intersecting $\c A$ either have total vertical displacement
		$>\frac vn$ or total horizontal displacement
		at least $v-\frac vn$.  Covered links are accounted for as before,
		so the result follows as in \eqref{eq perim bound}.
	\end{proof}
	
	
	
	\subsection{Convergence}
	\label{ssec conv}
	
	We now turn to the question of convergence of the measures 
	$\PP^\alpha_{\c D^\alpha_k,n,u}$ and  $\PP_{\Lambda_L,\beta,n,u}$
	(with $L\in 2\bb Z+\alpha$),
	particularly proving Lemma \ref{lem TV} and
	Theorem \ref{qrepair:thm:convergence}.
	As previously, we give the details for the
	case $\alpha=1$.  To lighten notation, we omit the subindex $\k$ from
	most notation.
	
	Let $\c D_1,\c D_2\subset\ZZ\times\RR$ be either large primal
	domains, or of the form $\Lambda_L\times[-\beta/2,\beta/2]$
	with $L\in 2\bb Z+\alpha$.  Let $\c B\se \c D_1\cap\c D_2$
	be a primal domain and $\c A\se \c B$
	a bounded domain, where we think of $\c B$ as much larger than
	$\c A$. 
	Write $\PP^1_{\c D_k,n,u}$ for the loop measure
	in $\c D_k$ (with primal or periodic boundary condition),
	$\PP^\otimes_{n,u}=\PP^1_{\c D_1,n,u}\otimes\PP^1_{\c D_2,n,u}$, 
	and $\ul{\om}=(\om_1,\om_2)$ for a sample of $\PP^\otimes_{n,u}$,
	so that $\om_1$ and $\om_2$ are independent random variables 
	with respective  laws $\PP^1_{\c D_1,n,u}$ and $\PP^1_{\c D_2,n,u}$.

	Write $\c P_k=\c P_k(\om_k)$ for the connected 
	component of small primal
	loops in $\om_k$ adjacent to $\partial \c D_k$, and write
	$\c E_k=\c D_k\setminus \c P_k$.  Note that $\c E_k$ is a union of
	(disjoint) connected sets $\c C(x_0,\om_k)$ 
	(as in Theorem \ref{thm hole perimeter}) for 
	various $x_0\in\c D_k$.
	Let $\c K_{\c A}=\c K_{\c A}(\ul\om)$ be the 
	connected component of 
	$\c A\cup\c E_1\cup \c E_2$ which contains $\c A$.
	Our main tool for proving convergence is the following:
	
	\begin{proposition}\label{qrepair:prop:convergence-prop}
		For each $u\in[0,1)$, there exists $\k_0,C,n_0>0$ such that for
		$\k\in[0,\k_0]$ and $n>n_0$, 
		\be
		\PP^\otimes_{n,u}[\c K_{\c A}\cap \c B^c \neq \varnothing]
		\le
		\vol(\c A)\e{-Cd(\c A,\c B^c)},
		\ee
		where $d=d_\infty$ the metric on $\ZZ\times\RR$ inherited from
		$\RR^2$.  The constants $C,n_0$ are uniform in the domains 
		$\c D_1,\c D_2,\c A,\c B$.
	\end{proposition}
	
	\begin{proof}
		The proof is a small modification of the proof of 
		Theorem \ref{thm hole perimeter} (which is in turn a small
		modification of the proof of Proposition \ref{prop:repair}).
		
		It is useful to think of the processes $\om_1$ and $\om_2$ on two
		separate copies of $\ZZ\times\RR$.  To that end, we work on
		$\ZZ\times\RR\times\{1,2\}$, and
		whenever we have sets $U_k\subset \c D_k$, perhaps dependent
		on $\om_k$, for $k=1,2$, we write $U_1\cupdot U_2 = (U_1\times\{1\})
		\cup (U_2 \times\{2\}) \subset \ZZ\times\RR\times\{1,2\}$. 
		
		Each of the  connected sets $\c C(x_0,\om_k)\se\c K_{\c A}$ has an
		\emph{outside} (as defined in the proof of  
		Theorem \ref{thm hole perimeter}) and 
		in this proof we write $\c O_k=\c O_k(\om_k)$ for the union of these
		outsides over all the  $\c C(x_0,\om_k)\se\c K_{\c A}$, 
		and we write $\c O=\c O(\ul\om)=\c O_1\cupdot\c O_2$.  
		We write $\ul{\om}^\ex=\om_1^{\ex}\cupdot\om_2^{\ex}$,
		where $\om_k^\ex$ are the exposed links of $\om_k$ lying in (or on
		the boundary of) $\c O_k$.
		We claim that it suffices to prove that there exist $C_2, C_3>0$
		such that for all $n>n_0$, 
		\be\label{q-repair-conv:eq:height-width-1}
		\begin{split}
			\PP^\otimes_{n,u}[\vol(\c O)>v] &\le  \vol(\c A)\e{-C_2nv},\\
			\PP^\otimes_{n,u}[|\ul{\om}^\ex|\ge vn, \vol(\c O)\leq v] 
			&\le \vol(\c A)\e{-C_3nv}.
		\end{split}
		\ee
		These two inequalities are analogous to \eqref{eq cluster 1} and 
		\eqref{eq cluster 2} in the proof of  Theorem 
		\ref{thm hole perimeter}, and as in the proof
		of Corollary \ref{component of a domain} they imply that 
		\be
		\PP^\otimes_{n,u}(\per(\c K_{\c A})>\per(\c A)+v)
		\leq \vol(\c A)\e{-C_4v}
		\ee
		for some $C_4>0$, from which the result follows.
		
		To prove \eqref{q-repair-conv:eq:height-width-1} we use an extension
		of the repair map $R$ to $\ul\om$.  We define this extension by
		applying the usual repair map in all of the 
		connected components $\c C(x_0,\om_k)\se\c K_{\c A}$.  As in the
		proofs of Proposition \ref{prop:repair} and 
		Theorem \ref{thm hole perimeter},
		to help count the number of
		preimages of the repair map we use a discretization into
		{blocks}.
		Formally, the blocks are the sets
		$(\{i,i+1\}\times[jh/n,(j+1)h/n]\times\{k\} )\cap \c D_k$, where 
		$i\in2\ZZ, j\in\ZZ, k\in\{1,2\}$, which are non-empty, and we 
		define the block-outside
		$\r b\c O$ as the union of those blocks which
		intersect $\c O$ non-trivially.   Blocks are now regarded as adjacent
		if they are either in the same copy $\ZZ\times\RR\times\{k\}$
		of $\ZZ\times\RR$ and are adjacent in the usual sense, or if they
		have the same $i$ and $j$ coordinates but differ in the $k$
		coordinate.  Thus each block is adjacent to (at most) 5 other blocks.  
		Similarly to the proof of Corollary \ref{component of a domain},
		we get a bound
		\be
		\frac{2\vol(\c A)}{2h/n} 25^m,
		\ee
		for the number of block-components $\r b_h\c O$ with $m$ blocks, 
		where the factor $2\vol(\c A)/\frac{2h}n$ accounts for the
		possible blocks intersecting $\c A$ (in either copy of
		$\ZZ\times\RR$).  Using this bound in place of 
		\eqref{eq graph theory 2}, the rest of the proof is the same as for
		Theorem \ref{thm hole perimeter}.
	\end{proof}
	
	Using Proposition \ref{qrepair:prop:convergence-prop} we can
	deduce Lemma \ref{lem TV}, the bound on the total
	variation distance between marginals:

\begin{proof}[Proof of Lemma  \ref{lem TV}]
  Write $U=\{\c K_{\c A}\cap \c B^c=\varnothing\}$, so 
  $\PP^{\otimes}_{n,u}[U^c]$ is bounded by Proposition 
  \ref{qrepair:prop:convergence-prop}.  
  Recall that $X$ is a bounded
  random variable depending only on the loops intersecting 
  $\c A$, and write 
  $\ul X=\ul X(\ul \om)=X(\om_1)-X(\om_2)$.  We claim that 
  $\EE^\otimes_{n,u}[\ul X\one_U]=0$,
  so that the expectations of $X$ under
  $\PP^1_{\c D_1,n,u}$ and $\PP^1_{\c D_2,n,u}$ differ by at most
  $2K\PP^\otimes_{n,u}[U^c]$,  where $K$ is an almost sure upper bound
  on $X$.  The claim then follows from Proposition
  \ref{qrepair:prop:convergence-prop}.  
  The reasoning is that on $U$
  there is a (random)
  primal circuit $\Gamma$ in $\c B$ surrounding $\c A$
  such that the marginal distributions of $\om_1$ and $\om_2$ in 
  $\c D_\Gamma$ are identical.
  See Figure \ref{fig marginal coupling} for an illustration of the argument that follows.  
  \begin{figure}[ht]
    \centering
    \includegraphics[scale=1]{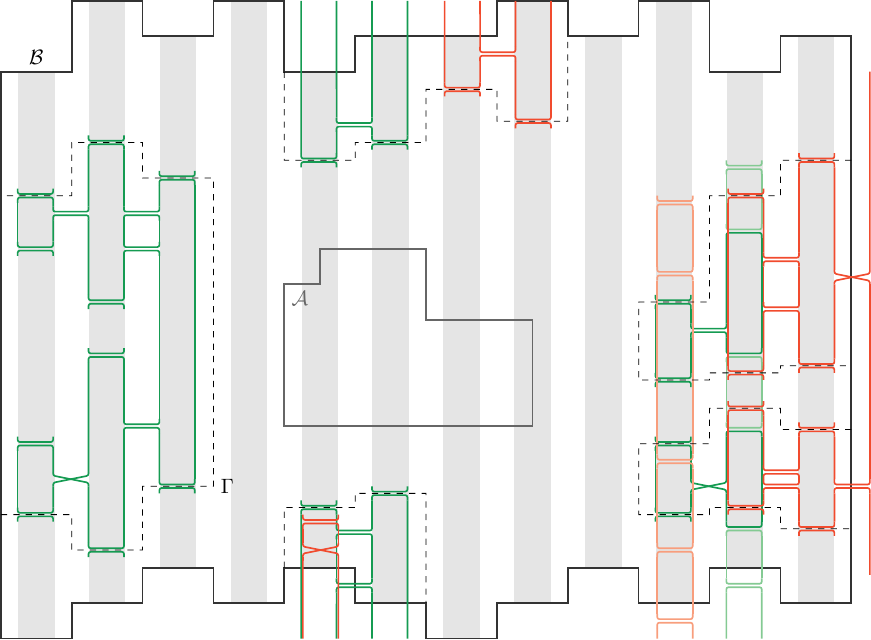}
    \caption{The primal domain 
      $\c B\se\c D_1\cap\c D_2$ (the domains $\c D_1\cap\c D_2$ not
      depicted) as well as the component $\c K_{\c B}'$, with $\om_1$ drawn
      green and $\om_2$ drawn orange.
      Long loops not part of $\c K_{\c B}'$, as well as most small loops,
      are not depicted.
      The exception is on the right part of the picture, 
      where some small loops are drawn in
      lighter colour.  The boundary curve $\Gamma$, where it deviates from
      $\partial \c B$, is drawn dashed.
    }\label{fig marginal coupling}
  \end{figure}
  
  To define $\Gamma$, let $\c K'_{\c B}$ be the connected component of 
  $\partial \c B\cup\c E_1\cup\c E_2$ which contains $\partial\c B$.
  We define $\c D_\Gamma:=\c B\setminus \c K'_{\c B}$.  Thus 
  each segment of $\Gamma$
  either belongs to $\partial \c B$, or to the boundary of a
  component $\c C(x_0,\om_k)$ for some $x_0\in\c B$.
  We notice the following properties of
  $\Gamma$.  First, each vertical segment of $\Gamma$ is in a dual
  column (white; odd left endpoint) so $\c D_\Gamma$ is indeed a primal
  domain.  Next, each vertical interval of $\Gamma$  
  traverses no links of $\om_1$  or of $\om_2$, since any such link
  would belong to a long loop or a dual loop which would then be part of
  $\c K'_{\c B}$.  Finally, each horizontal segment of $\Gamma$
  either belongs to  $\partial \c B$, or
  traverses a primal double-bar $\dbar$ of exactly one of 
  $\om_1$ and $\om_2$.  Moreover, if this double-bar belongs to $\om_1$
  then it lies in a small primal loop of $\om_2$, and vice versa.  
  
		We write $\c D_\Gamma'$ for $(\c D_1\cup\c D_2)\sm \c D_\Gamma$
		and $\eta_1,\eta_2$ for the restrictions of $\om_1,\om_2$
		to $\Gamma\cup\c D'_\Gamma$, and $\gamma_1,\gamma_2$ for the 
		restrictions of $\om_1,\om_2$ to $\Gamma$.
		We also write $\gamma=\gamma_1\cup\gamma_2$ for the set of links on
		$\Gamma$ in either configuration.   The key claim is that  if
		we were to modify $\om_1$ by including the links of $\gamma_2$, then
		the number of loops would change by a term which depends only on $\eta_1$
		and $\eta_2$, i.e.\ only on the configuration outside $\c D_\Gamma$ 
		(and similarly for including the links of $\gamma_1$ in
		$\om_2$).  Thus, if we condition on $\eta_1,\eta_2$, 
		thereby regarding
		them as fixed, then this change is deterministic, and up to this
		deterministic change the number of loops in $\c D_\Gamma$  is counted
		according to the primal boundary condition.  In particular, this gives
		the same boundary condition for both 
		$\om_1$ and $\om_2$ so they have the same conditional distribution. 
		
		To make a  precise formulation of the above claim, define
		$\overline\gamma_k$, $k\in\{1,2\}$,
		by adding to $\gamma_k$ a double-bar at each
		horizontal segment of $\Gamma$ where it
		coincides with $\partial \c B$ and traverses a primal column,
		and write $\ol\gamma=\ol\gamma_1\cup\ol\gamma_2$.
		Then we have
		\be\label{eq loop claim}
		\ell(\om_k)=\ell(\om_k\cup\ol\gamma)
		-|\ol\gamma_{3-k}|, 
		\qquad k\in\{1,2\}.
		\ee
		To justify \eqref{eq loop claim}, for simplicity take $k=1$, and note
		that (due to our observations about $\Gamma$ above)
		any loop of $\om_1$ that intersects both $\c D_\Gamma$ and
		$\c D_\Gamma'$ is necessarily a small primal loop which traverses 
		some number $m\geq1$ of double-bars of $\ol\gamma_2$.   
		When adding the
		links of $\ol\gamma_2$ to $\om_1$, this small primal loop is 
		replaced by $m$ small primal loops.   Thus, each link of $\ol\gamma_2$
		contributes exactly one extra loop under the modification
		$\om_1\mapsto\om_1\cup\ol\gamma$.
		
		To make the rest rigorous 
		we use Mecke's formula, Lemma \ref{lem mecke}.   We can write
		\be\begin{split}
			\EE^\otimes_{n,u}[\ul X\one_U]=
			\frac{1}{Z_{\c D_1}Z_{\c D_2}}\sum_{r_1,r_2\geq 0}
			\EE_1\otimes\EE_1\Big[&
			\sum_{\substack{\eta_1\se\om_1\\|\eta_1|=r_1}}
			\sum_{\substack{\eta_2\se\om_2\\|\eta_2|=r_2}}
			\one\{\om_k\cap(\Gamma\cup\c D'_\Gamma)=\eta_k,\;k=1,2\}\\
			&\quad\one_U(\om_1,\om_2)
			\ul X(\om_1,\om_2)
			n^{\ell(\om_1)+\ell(\om_2)}
			\Big],
		\end{split}\ee
		where the last expectation $\EE_1\otimes\EE_1[\dotsb]$ can be written 
		as
		\be\begin{split}\label{eq Gamma ce}
			\int\dd\mu^{\odot r_1}(\eta_1) \int\dd\mu^{\odot r_2}(\eta_2)
			\EE_1\otimes\EE_1[&
			\one\{(\om_k\cup\eta_k)\cap(\Gamma\cup\c D'_\Gamma)=\eta_k,\;k=1,2\}
			\one_U(\om_1\cup\eta_1,\om_2\cup\eta_2)\\
			&\quad
			\ul X(\om_1\cup\eta_1,\om_2\cup\eta_2)
			n^{\ell(\om_1\cup\eta_1)+\ell(\om_2\cup\eta_2)}].
		\end{split}\ee
		In this expression, note that $\om_1,\om_2$ are configurations in 
		$\c D_\Gamma$ constrained to belong to the event
		$V$ that in $\om_k\cup\ol\gamma$,
		only small primal loops are adjacent to $\Gamma$. 
		Since $U$ equals the event that $\c K'_{\c B}\cap\c A=\es$,
		it depends only on $\eta_1,\eta_2$, and similarly $\ul X$
		depends only on $\om_1,\om_2$.  The weights 
		$\w(\om_k\cup\eta_k)$ factorize over $\om_k$ and $\eta_k$, 
		and \eqref{eq loop claim} can be written as
		\be
		\ell(\om_k\cup\eta_k)=\ell_{\c D_\Gamma}^1(\om_k)
		+\ell_{\c D'_\Gamma}^1(\eta_k)
		-|\ol\gamma_{3-k}|, 
		\qquad k\in\{1,2\},
		\ee
		where $\ell_{\c D_\Gamma}^1(\om_k)$ 
		and $\ell_{\c D'_\Gamma}^1(\eta_k)$ count the number of 
		loops 
		with primal boundary condition.  Taken
		together, this means that the expectation in \eqref{eq Gamma ce}
		can be factorized as
		\be\label{eq Gamma ce 2}
		\EE_1\otimes\EE_1[\ul X(\om_1,\om_2)
		n^{\ell_{\c D_\Gamma}^1(\om_1)+\ell_{\c D_\Gamma}^1(\om_2)}
		\one_V(\om_1)\one_V(\om_2)]
		\cdot F(\eta_1,\eta_2)
		\ee
		for some function $F$.  Recalling that
		$\ul X(\om_1,\om_2)=X(\om_1)-X(\om_2)$, 
		it follows that the expectation in \eqref{eq Gamma ce 2} is in fact
		identically $=0$.   Thus
		$\EE^\otimes_{n,u}[\ul X\one_U]=0$ as claimed, and Lemma \ref{lem TV}
		is proved.
	\end{proof}
	
	Before we turn to the proof of Theorem 
	\ref{qrepair:thm:convergence} we introduce
	the precise notion of Gibbs measures for the loop model.  
	Recall that we identify link-configurations $\omega$ with
	counting-measures on 
	$(\bb Z+\tfrac12)\times\bb R\times\{\cross,\dbar\}$.
	This applies both to configurations in infinite volume and
	to configurations in a bounded
	domain $\c D$.  In the latter case $\omega$ simply has no links 
	outside $\c D$ or traversing the boundary $\partial\c D$
	(recall that we have defined domains $\c D$ as open sets, so links on
	$\partial \c D$ by definition do not lie in $\c D$).
	
	Let $\c D\se\bb Z\times\bb R$ be a bounded 
	domain, not necessarily primal or dual.  Any link-configuration 
	$\tau$ in $(\bb Z\times\bb R)\sm\c D$ imposes a boundary condition on
	the loops inside $\c D$ as follows (see Figure \ref{fig gibbs} for an
	illustration). First,
	define the `horizontal boundary'
	$\partial_{\r h}\c D$ of $\c D$ as the set of points on
	$\partial \c D$ which are of the form $(x,t)$ with $x\in\bb Z$
	(these are necessarily on the `top and bottom' of $\c D$)
	and the `vertical boundary' $\partial_{\r v}\c D$ as the union of
	vertical intervals forming $\partial \c D$.
	To any link of $\tau$ on (i.e.\ crossing) the vertical boundary 
	$\partial_{\r v} \c D$, say at height $t$, corresponds 
	two points inside $\c D$ at heights $t\pm 0$, i.e.\ the two endpoints
	of the link in the domain.
	We define $\partial_\tau\c D$ as the collection of such 
	points together with the horizontal boundary $\partial_{\r h}\c D$.
	Then, the configuration $\tau$ defines a 
	\emph{partial pairing} of $\partial_{\tau}\c D$, where two points are
	paired if they are connected by a loop-segment of $\tau$ lying
	entirely outside $\c D$.  The pairing is only partial due to the
	possible existence of multiple unbounded segments.  
	We define a loop-measure on link-configurations in $\c D$
	by
	\be
	\PP_{\c D,n,u}^\tau(A)=\frac1{Z^{\tau}_{\c D,n,u}}
	\int  \one_A\dd \PP_1(\omega) \; n^{\ell(\omega;\tau)},
	\ee
	where $\ell(\omega;\tau)$ is the number of loops in $\c D$
	counted according to the boundary condition above.
	This definition includes the cases of primal, dual
	and periodic boundary conditions
	\eqref{eq measure periodic} 
	by appropriate choice of $\tau$.
	
	\begin{figure}[ht]
		\centering
		\includegraphics[scale=1]{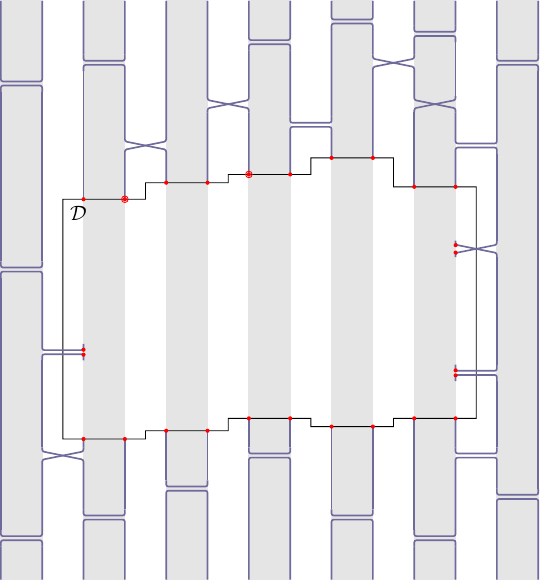}
		\caption{A domain $\c D$ and a link-configuration $\tau$ outside 
			$\c D$, including some links crossing $\partial_{\r v} \c D$.
			The configuration $\tau$ defines a
			partial pairing of $\partial_\tau\c D$, the latter illustrated using red
			dots.  Two points of $\partial_\tau\c D$, on the top boundary of 
			$\c D$ (highlighted), are unpaired.
		}\label{fig gibbs}
	\end{figure}

	Define $\c F_{\c D}$ and  $\c F_{\c D^c}$ as the $\sigma$-algebras of
	events depending on $\omega\cap\c D$ and on 
	$\omega\cap\c D^c$, respectively.  A probability
	measure $\PP$ on link-configurations in
	$\bb Z\times\bb R$ is called a \emph{Gibbs-measure}
	if for all bounded rectangular 
	domains $\c D\se\bb Z\times\bb R$,
	\be\label{eq gibbs}
	\PP(\cdot\mid \c F_{\c D^c})(\tau)=\PP_{\c D,n,u}^\tau(\cdot),
	\qquad\text{for $\PP$-a.e. }\tau.
	\ee
	We use a similar definition on partly infinite domains 
	$\bb Z\times[-\beta/2,\beta/2]$ (periodic in the second coordinate) 
	and $\{-K+1,\dots,L\}\times\bb R$.
	
	\begin{lemma}\label{lem Gibbs}
		Let $\c D_k$ be a sequence of domains
		with $\c D_k\nearrow \ZZ\times\RR$
		or $\bb Z\times[-\beta/2,\beta/2]$ or 
		$\{-K+1,\dots,L\}\times\bb R$.
		Let $\tau_k$ be a sequence of link-configurations,
		and let
		$\PP$ be  a subsequential limit of 
		$\PP^{\tau_k}_{\c D_k,n,u}$ as $k\to\oo$.
		If $\PP$ is supported on configurations with at most one infinite
		loop, then $\PP$ is a Gibbs measure. 
		
		In particular, if $\c D_k^\alpha$ is a sequence of primal (for $\alpha=1$)
		or dual (for $\alpha=2$) domains
		with $\c D_k^\alpha\nearrow \ZZ\times\RR$, and 
		$\Lambda_L=\{-L+1,\dots,L\}\subset\ZZ$ with 
		$L\in2\mathbb Z+\alpha$, then any subsequential limit $\PP$
		of $\PP^\alpha_{\c D^\alpha_k,n,u}$ or $\PP_{\Lambda_L,\beta,n,u}$
		is a Gibbs measure.
	\end{lemma}
	\begin{proof}
		The second claim follows from the first, since
		Theorem \ref{thm hole perimeter} implies that
		any subsequential limit of $\PP^\alpha_{\c D^\alpha_k,n,u}$ 
		or $\PP_{\Lambda_L,\beta,n,u}$
		has \emph{no} infinite loop, almost surely.
		Hence we focus on the first claim.
		
		Fix a bounded domain $\c D$,
		let $A\in\c F_{\c D}$ and 
		let $\c D_m$ and $\tau_m$ be such that 
		$\PP^{\tau_m}_{\c D_m}\Rightarrow\PP$.
		The key observation is that, for configurations 
		$\tau$ with at most one infinite loop, there is a number
		$k_0(\tau)<\oo$ such that
		\be\label{eq gibbs key claim}
		\PP^{\tau_m}_{\c D_m}(A\mid \c F_{\c D_k\sm\c D})(\tau)
		=\PP^\tau_{\c D}(A),
		\qquad\text{for }m>k>k_0(\tau).
		\ee
		Indeed, for $k$ large enough, any finite loop-segment
		connecting points of $\partial_{\bb Z}\c D$ is entirely contained in
		$\c D_k$, leaving at most two points which must then both lie on the unique
		infinite loop.  It follows that for such $k$,
		the partial
		pairing of $\partial_{\bb Z}\c D$ defined by $\tau$ is determined
		within $\c D_k$, which implies \eqref{eq gibbs key claim}.
		
		Using \eqref{eq gibbs key claim}, we will show that
		\be\label{eq ce conv}
		\PP(A\mid \c F_{\c D^c})(\tau)
		=\PP^\tau_{\c D}(A),
		\qquad \PP\text{-almost surely in }\tau,
		\ee
		which is indeed the condition \eqref{eq gibbs}
		for $\PP$ to be a Gibbs measure.
		To see  \eqref{eq ce conv}, fix $\ell>0$ and
		let $X=X(\tau)$ be 
		bounded and $\c F_{\c D_\ell\sm\c D}$-measurable.  Then
		for any $k>\ell$,
		\be
		\EE[\PP^\tau_{\c D}(A)X(\tau)]=
		\EE[\PP^\tau_{\c D}(A)X(\tau) \one\{k>k_0(\tau)\}]+
		\EE[\PP^\tau_{\c D}(A)X(\tau) \one\{k\leq k_0(\tau)\}].
		\ee
		(Here and in what follows the outermost
		expectation is over the configuration $\tau$.)
		The second term goes to 0 as $k\to\oo$, while
		by the assumption $\PP^{\tau_m}_{\c D_m}\Rightarrow\PP$, 
		the first term satisfies
		\[\begin{split}
			\EE[\PP^\tau_{\c D}(A)X(\tau) \one\{k>k_0(\tau)\}]
			&=
			\lim_{m\to\oo}
			\EE^{\tau_m}_{\c D_m}[\PP^\tau_{\c D}(A)X(\tau)\one\{k>k_0(\tau)\}]\\
			&=\lim_{m\to\oo}
			\EE^{\tau_m}_{\c D_m}[
			\PP^{\tau_m}_{\c D_m}(A\mid \c F_{\c D_k\sm\c D})(\tau)
			X(\tau)\one\{k>k_0(\tau)\}],
			\quad\text{ by \eqref{eq gibbs key claim}},\\
			&=\lim_{m\to\oo}
			\EE^{\tau_m}_{\c D_m}[
			\PP^{\tau_m}_{\c D_m}(A\mid \c F_{\c D_k\sm\c D})(\tau)
			X(\tau)]\\
			&\quad\qquad-\EE^{\tau_m}_{\c D_m}[
			\PP^{\tau_m}_{\c D_m}(A\mid \c F_{\c D_k\sm\c D})(\tau)
			X(\tau)\one\{k\leq k_0(\tau)\}]\\
			&=\lim_{m\to\oo}
			\EE^{\tau_m}_{\c D_m}[\one_A X(\tau)]+o(1)\\
			&=\EE[\one_A X(\tau)]+o(1).
		\end{split}\]
		Here the $o(1)$-term vanishes as $k\to\oo$ and
		we used the 
		$\c F_{\c D_k\sm\c D}$-measurability of $X(\tau)$.
		Thus $\EE[\PP^\tau_{\c D}(A)X(\tau)]=\EE[\one_A X(\tau)]$
		for all bounded and
		$\c F_{\c D_\ell\sm\c D}$-measurable $X$, for all $\ell>0$,
		hence the same is true for all $\c F_{\c D^c}$-measurable $X$
		(by the $\pi$-$\lambda$-theorem).
		Since $\PP^\tau_{\c D}(A)$ is $\c F_{\c D^c}$-measurable, 
		\eqref{eq ce conv} follows. \end{proof}
	

	\begin{proof}[Proof of Theorem \ref{qrepair:thm:convergence}]
		We focus on the case $\alpha=1$ as the case $\alpha=2$ is the same.  
		Recall that we work on the measurable space
		$(\c M^\#, \c B(\c M^\#))$ 
		where $\c M^\#$ is the set of boundedly finite measures on
		$(\ZZ+\tfrac12)\times\RR\times\{\cross,\dbar\}$
		and $\c B(\c M^\#)$ the natural  Borel $\sigma$-algebra.
		We first note that our collections of measures 
		$\PP^1_{\c D_k,n,u}$ or $\PP_{\Lambda_L,\beta,n,u}$
		are uniformly tight, indeed by 
		\cite[Proposition 11.1.VI]{daley-verejones-2} it suffices to check that
		given any compact set $\c K\se (\ZZ+\tfrac12)\times\RR$
		and any $\eps>0$, the probability that $\c K$ contains more than $M$
		links is uniformly $<\eps$ for $M$ large enough, which in our
		case is obvious since our measures are stochastically dominated by
		Poisson processes.  Thus, it suffices to establish uniqueness of
		subsequential limits, for which in turn it suffices
		by \cite[Corollary 9.2.IV]{daley-verejones-2}
		to establish that
		the finite-dimensional marginals are uniquely determined.
		In the  case of primal domains $\c D_k$
		and the measures $\PP^1_{\c D_k,n,u}$, this is an immediate
		consequence of Lemma \ref{lem TV}.  The same result
		moreover implies that the limit does not depend on the choice of
		sequence of primal domains.
		The limits are Gibbs measures by Lemma \ref{lem Gibbs}.

		For the case of $\PP^\p_{\Lambda_L,\beta,n,u}$
		when $\beta\to\oo$ followed by $L\to\oo$ 
		(with $L\in 2\bb Z+1$), Lemma \ref{lem TV} is not immediately
		applicable (Mecke's formula is not available after taking 
		$\beta\to\oo$) 
		so we argue slightly differently.  
		Let $\c A$ be an arbitrary bounded domain and
		let $\PP_{L,\oo}=\PP^\p_{\Lambda_L,\oo,n,u}$ be a
		subsequential limit when $\beta\to\oo$, where $L$ is large enough that
		$\c A\se \c R_L:=[-L/2,L/2]^2\cap (\ZZ\times\RR)$.
		Let $U_L$ be the event that 
		$\c R_L$ is surrounded by a circuit of small
		primal loops all of whose points are outside $\c R_L$ and at vertical
		height at most $\pm L$.
		By Corollary \ref{component of a domain} 
		$\PP_{L,\oo}(U_L^c)$ decays exponentially in $L$
		(with a polynomial prefactor). 
		On $U_L$ we let $\gamma$ be the outermost choice of
		such a circuit of loops (which can be found by exploring from 
		$[-L,L]^2$ inwards), we let $\xi$ be the links on $\gamma$, and we
		let $\c D_\xi$ be the primal domain containing $\c A$ which is
		delimited by $\gamma$.
		Since $\PP_{L,\oo}$  is a Gibbs-measure, by Lemma 
		\ref{lem Gibbs}, the conditional distribution in
		$\c D_{\xi}$ given $\xi$ and the configuration outside
		$\c D_{\xi}$ is  $\PP^1_{\c D_\xi,n,u}$.
		See Figure \ref{fig circuit beta=oo}.
		
		\begin{figure}[ht]
			\centering
			\includegraphics[scale=.9]{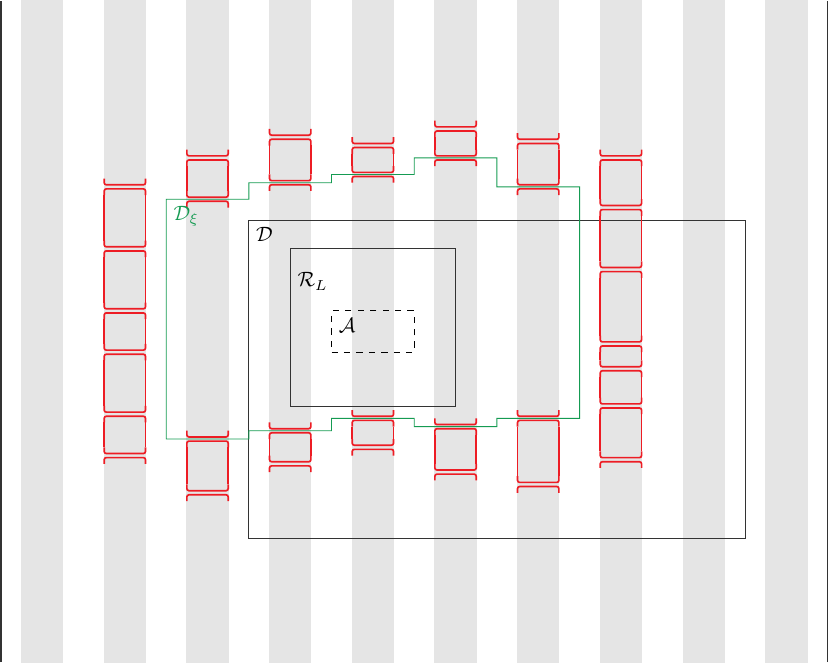}
			\caption{
				Illustration of part of a sample of $\PP_{L,\oo}$.
				The rectangle $\c R_L\supseteq \c A$ is exponentially 
				likely (in $L$)  to be surrounded by a circuit of small primal loops, and
				since $\PP_{L,\oo}$ is a Gibbs measure, the conditional distribution
				inside that circuit is $\PP^1_{\c D_\xi,n,u}$.  The marginal
				distribution inside $\c A$ is then exponentially close to that of 
				$\PP^1_{\c D,n,u}$ for any other primal domain $\c D$ 
				containing $\c R_L$.
			}\label{fig circuit beta=oo}
		\end{figure}
		
		Now let  $\c D$ be any   primal domain containing $\c R_L$. 
		For any event $A$  depending only on the 
		configuration of links in $\c A$, by Lemma \ref{lem TV}
		with $\c D_1=\c D_\xi$, $\c D_2=\c D$, and
		$\c B=\c D_1\cap\c D_2$ we have for some $C>0$ that
		\be
		|\PP_{\c D_\xi,n,u}^1(A)-\PP_{\c D,n,u}^1(A)|
		\leq \e{-CL}.
		\ee
		Then
		\be\begin{split}
			|\PP_{L,\oo}(A)-\PP_{\c D,n,u}^1(A)|&\leq
			\PP_{L,\oo}(U_L^c)+
			|\EE_{L,\oo}[\PP_{L,\oo}(A\mid\c F_{\c D_\xi^c})\one_U]-\PP_{\c D,n,u}^1(A)|=
			\\
			&\leq \PP_{L,\oo}(U_L^c)+
			\EE_{L,\oo}[|\PP_{\c D_\xi,n,n}^1(A)\one_{U_L}
			-\PP_{\c D,n,u}^1(A)|]\\
			&\leq \e{-C' L}
		\end{split}\ee
		for some $C'>0$.  
		Letting $L\to\oo$ 
		this gives that any subsequential
		limit as $L\to\oo$ coincides with the limit obtained above using
		primal domains $\c D_k\nearrow\bb Z\times\bb R$.
		A similar argument works for the case when $\beta,L\to\oo$
		simultaneously and for the case of the domains
		$\c D_{L,\beta}$ with any order
		of limits.

		The $2\ZZ\times\RR$-invariance follows from the independence of the
		choice of domains, and the fact that
		$\tau_{(1,0)}\PP^\alpha_{n,u}=\PP^{3-\alpha}_{n,u}$ is clear.  
		Theorem \ref{thm hole perimeter} extends to the infinite volume measure
		$\PP^1_{n,u}$ to show that it is supported on configurations with no
		infinite clusters of $\c E_1=\c P^c$ where $\c P$ is the union of
		unbounded components of small primal loops. 
		The corresponding statement follows for
		$\PP^2_{n,u}$, and it also follows that $\PP^1_{n,u}$ and
		$\PP^2_{n,u}$ are distinct.  
		
		The proof that $\PP^1_{n,u}$ and $\PP^2_{n,u}$ are ergodic is very
		similar to the proof of decay of correlations in 
		Theorem \ref{thm:dimersation} so we only give an outline,
		and we focus on the case $\PP^1_{n,u}$.    Let 
		$\c D_1$ and $\c D_2$ be disjoint domains, which we think of as far
		apart, let $A\in\c F_{\c D_1}$, $B\in\c F_{\c D_2}$,
		and let $\c D$ be a primal domain containing both 
		$\c D_1$ and $\c D_2$.
		The argument in Theorem \ref{thm:dimersation} shows that,  
		under $\PP^1_{\c D,n,u}$, the domains $\c D_1$ and $\c D_2$ 
		are very likely to be separated by a circuit of small primal loops.
		It follows that 
		$|\PP^1_{\c D,n,u}(A\cap B)-\PP^1_{\c D,n,u}(A) \PP^1_{\c D,n,u}(B)|$
		decays exponentially in the distance between 
		$\c D_1$ and $\c D_2$, uniformly in $\c D$.
		Thus $\PP^1_{n,u}$ is mixing under $2\ZZ\times\RR$-shifts,
		\be
		\lim_{|k|+|t|\to\oo} 
		\PP^1_{n,u}(A\cap \tau_{(2k,t)}^{-1} B)=
		\PP^1_{n,u}(A) \PP^1_{n,u}(B),
		\qquad \text{for all } A,B\in\c F,
		\ee
		and hence ergodic.
	\end{proof}

	\section*{Declarations}
	
	\subsection*{Funding}
	
	The research of JEB was supported by 
	Vetenskapsr{\char229}det, 
	grant 2023-05002, 
	and by \emph{Ruth och Nils Erik Stenbäcks stiftelse}. 
	KR was supported by the FWF Standalone grants ``Spins, Loops and
	Fields'' P 36298 and ``Order-disorder phase transition in 2D lattice
	models'' P 34713,  the FWF SFB Discrete Random Structures F 1002, and
	the Academy of Finland Centre of Excellence Programme grant number
	346315 ``Finnish centre of excellence in Randomness and STructures''
	(FiRST).  
	
	\subsection*{Competing interests}
	
	The authors have no relevant financial or non-financial interests to
	disclose.  
	The authors have no competing interests to declare that are relevant
	to the content of this article.
	All authors certify that they have no affiliations with or involvement
	in any organization or entity with any financial interest or
	non-financial interest in the subject matter or materials discussed in
	this manuscript. 
	The authors have no financial or proprietary interests in any material discussed in this article.
	
	\subsection*{Data availability}
	
	This manuscript has no associated data.

\end{document}